\documentclass{article}[11pt] 
\usepackage{fullpage}
\usepackage{authblk}

\usepackage[hyphens]{url}  
\usepackage{graphicx} 
\usepackage{caption} 
\usepackage{algorithm}
\usepackage{algorithmic}

\usepackage{newfloat}
\usepackage{listings}
\DeclareCaptionStyle{ruled}{labelfont=normalfont,labelsep=colon,strut=off} 
\floatstyle{ruled}
\newfloat{listing}{tb}{lst}{}
\floatname{listing}{Listing}
\title{Stochastic Multi-round Submodular Optimization with Budget\thanks{This work is partially supported by: GNCS-INdAM; the PNRR MIUR project FAIR - Future AI Research (PE00000013), Spoke 9 - Green-aware AI; the Project SERICS (PE00000014) under the NRRP MUR program funded by the EU – NGEU; MUR - PNRR IF Agro@intesa; the Italian MIUR PRIN 2017 project ALGADIMAR - Algorithms, Games, and Digital Markets (2017R9FHSR\_002).}}
\author[1]{ Vincenzo Auletta\thanks{auletta@unisa.it}}
\author[2]{ Diodato Ferraioli\thanks{dferraioli@unisa.it}}
\author[3]{Cosimo Vinci\thanks{cosimo.vinci@unisalento.it (corresponding author)}}
\affil[1,2]{Department of Information Engineering, Electrical Engineering and Applied Mathematics, University of Salerno, Italy.}
\affil[3]{Department of Mathematics and Physics ``Ennio De Giorgi'', University of Salento, Italy.}
\date{}
\def\RP{{\mathbb{R}}_{\geq 0}}

\newcommand{\fT}{\bar{f}}
\newcommand{\SGR}{SingleGr}
\newcommand{\MGR}{MultiGr}
\def\SRM{\text{SRm}(t,b)}

\usepackage{algorithm}
\usepackage{algorithmic}

\usepackage{bm}
\usepackage{amsmath,amsthm,stmaryrd}
\usepackage{amsfonts}

\newtheorem{theorem}{Theorem}
\newtheorem{lemma}{Lemma}
\newtheorem{corollary}{Corollary}
\newtheorem{proposition}{Proposition}
\theoremstyle{definition}
\newtheorem{definition}{Definition}
\newtheorem{remark}{Remark}

\allowdisplaybreaks

\usepackage{xcolor}

\begin{document}

\maketitle

\begin{abstract}
In this work, we study the Stochastic Budgeted Multi-round Submodular Maximization (SBMSm) problem, where we aim to adaptively maximize the sum, over multiple rounds, of a monotone and submodular objective function defined on subsets of items. The objective function also depends on the realization of stochastic events, and the total number of items we can select over all rounds is bounded by a limited budget.
This problem extends, and generalizes to multiple round settings, well-studied problems such as (adaptive) influence maximization and stochastic probing.

We show that, if the number of items and stochastic events is somehow bounded, there is a polynomial time dynamic programming algorithm for SBMSm. Then, we provide a simple greedy $1/2(1-1/e-\epsilon)\approx 0.316$-approximation algorithm for SBMSm, that first non-adaptively allocates the budget to be spent at each round, and then greedily and adaptively maximizes the objective function by using the budget assigned at each round. Finally, we introduce the {\em budget-adaptivity gap}, by which we measure how much an adaptive policy for SBMSm is better than an optimal partially adaptive one that, as in our greedy algorithm, determines the budget allocation in advance. We show that the budget-adaptivity gap lies between $e/(e-1)\approx 1.582$ and $2$.


\end{abstract}

{\normalfont \bf Keywords:} Combinatorial Optimization; Stochastic Adaptive Optimization; Submodularity.

\section{Introduction}
Consider the well-known problem of influence maximization: here we are given a (social) network, with edge probabilities encoding the strength of the relationships among network's components; these probabilities describe how easily influence (e.g., an information or an innovation) spreads over the network starting from a limited set of nodes, usually termed \emph{seeds}; the goal is then to find the set of seeds that maximizes the expected number of ``influenced'' nodes.
The non-adaptive version of the problem, in which all the seeds are chosen before that influence starts, has been widely studied, with numerous successful algorithms \cite{kempe2003maximizing,chen2009approximability,borgs2014maximizing} and applications ranging from viral marketing \cite{kempe2003maximizing}, to election manipulation 
\cite{wilder2018controlling,coro2019exploiting,castiglioni2021election}, campaigns to contrast misinformation \cite{AmorusoAACFR20} to infective disease prevention \cite{yadav2017influence}.
Recently, focus moved to the adaptive variant of this problem (in which seeds can be added even when the influence is spreading from previously placed seeds) \cite{golovin2011adaptive,peng2019adaptive,d2021better}. This variant happens to have strong relationship with some \emph{Stochastic Probing} problems \cite{asadpour2008stochastic,asadpour2016maximizing}, in which there are $n$ items whose states are stochastically defined according to (not necessarily equal) independent distributions and the goal is to choose a feasible subset of items maximizing the expected value of a non-negative, monotone and submodular set function. Stochastic probing is relevant in various applications, such as the optimal placement of sensors for environmental monitoring \cite{DBLP:conf/uai/KrauseG05} or energy resources within electrical distribution networks \cite{OLADEJI2022100897,9543204}. In these scenarios, items represent potential locations within a network, where each item's state models the (uncertain) efficiency of a sensor or energy resource placed in the corresponding location. Feasible solutions correspond to sets of locations where sensors or resources can be deployed and/or activated, subject to cost or pollution constraints. The submodular set function models either the coverage provided by the sensors or the energy efficiency achieved. Other stochastic probing settings can be also encountered in expert hiring \cite{hoefer2021stochastic}, online dating and kidney exchange \cite{gupta2013stochastic,Bradac0Z19}, Bayesian mechanism design and sequentially posted pricing \cite{chawla2010multi,adamczyk2017sequential}.

Research both in influence maximization and stochastic probing focused almost exclusively on the single-round setting. Recently, there is an increasing interest in those cases in which the optimization is run in a multi-round setting. This is highlighted both by the increasing amount of literature on the topic (see Related Works), and by the numerous applications that are evidently deployed over multiple rounds, as it is the case, e.g., for many real-word marketing,   electoral and misinformation/disinformation contrasting campaigns. Indeed, it is often the case that in these contexts we are interested both in maximizing the total effectiveness of the campaign, and in limiting the total budget spent. Note that, in order to achieve this goal it may be insufficient to limit the budget spent at each round, since this may lead in general to very inefficient outcomes (see Remark \ref{remalimit} below for more details). However, this makes the problem much more complex: we do not only need to care about which actions to take at each round, but also we need to care about how large must be the portion of the budget spent at that round. Similarly, in contexts involving the optimal use of sensors or energy resources over time, it's possible that the cost or pollution constraints are not imposed daily, but rather on an average basis over the entire time-horizon. This means that while certain days may exceed these limits, the overall usage throughout the planning period must meet the constraints on average.

In this work, we aim to extend the study of above cited problems (and their generalizations) by jointly considering all these properties: multiple-rounds, total budget, and adaptivity.

\paragraph{Related Works.}
Previous works indeed never consider all these three properties jointly.

The interest in multi-round variants of the above discussed problems is testified by the increasing amount of papers focusing on this aspect: for example, \cite{sun2018multi} considered an influence maximization problem over multiple and independent rounds of diffusion, where the goal is to maximize the number of nodes selected in at least one round, subject to the selection of at most $k$ seeds for each round (thus, in contrast to our setting, they have a limit only on per-round budget); in the work of \cite{tong2020time}, instead, a round denotes a time step in which either a node is selected, or the already active nodes influence some of their neighbors, and the goal is to select at most $k$ seeds that maximize the influence spread within $T$ rounds.

Moreover, running a campaign over multiple rounds can be very useful (maybe necessary) to the designer of the campaign in order to gain the necessary knowledge about the environment to address the desired goal. This approach has been widely studied under the lens of online learning \cite{gittins1979bandit}.
Specifically, online learning algorithms for some submodular maximization processes (including the simplest form of our multi-round influence maximization problem) have been proposed \cite{vaswani2015influence,chen2016combinatorial,wen2017online}. In particular, these works introduced an approach, named \emph{Combinatorial Upper Confidence Bound} (CUCB), that extended the well-known UCB bandit algorithm \cite{AuerCF02} with the aim of designing online learning algorithms with good guarantees for a wide class of optimization problems. However, such approach has been defined only for non-adaptive optimization frameworks, e.g., standard influence maximization problems in which the diffusion is not observed after each seed selection, but only at the end of the selection process.  \cite{gabillon2013adaptive} presented CUCB algorithms for a class of adaptive optimization problems simpler than the one considered in this work, since, e.g., the seeding is subject to a limit on per-round budget and the underlying combinatorial structure is not sufficiently rich to represent influence maximization problems. Budget constraints without any limit on per-round quotas have been considered by \cite{das2022budgeted}, that still achieves good guarantees, but considers an even simpler setting than the one by  \cite{gabillon2013adaptive}, that in fact does not apply to influence maximization.

The analysis of optimization advantages of adaptive w.r.t. non-adaptive algorithms has been initiated by \cite{dean2008approximating,dean2005adaptivity} on classical packing problems, mainly trough the {\em adaptivity gap}, a metric defined as the worst-case ratio between the value achieved by an optimal adaptive algorithm and the optimal non-adaptive value. Only recently, this kind of analysis has been applied to stochastic maximization problems, including \emph{single-round} influence maximization and stochastic probing. In particular,  \cite{golovin2011adaptive} showed that whenever the objective function satisfies a property named \emph{adaptive submodularity} (see Definition~\ref{def:ad_subm} below), an approximate greedy algorithm achieves a $(1-1/e-\epsilon)$-approximation of the optimal greedy algorithm, with $e\approx 2.718$ denoting the Nepero number (that is, $1-1/e\approx 0.632$) and $\epsilon>0$ being an arbitrarily small constant.  \cite{peng2019adaptive,ChenP19} also proved constant bound on the adaptivity gap for some influence maximization problems, i.e., by using a non-adaptive optimal algorithm we lose only a constant approximation factor w.r.t. the optimal adaptive algorithm. Some bounds on the approximation ratio of adaptive algorithms and on the adaptivity gap have been recently improved by  \cite{d2021better,d2021improved,ChenPST22}. The adaptivity gap of  stochastic probing problems has been also extensively studied  \cite{asadpour2008stochastic,asadpour2016maximizing,GuptaNS16,GuptaNS17,Bradac0Z19}.

While our model and results are clearly inspired by these works, machineries defined therein cannot be directly applied to our setting since we are also interested in allocating the budget among multiple rounds (for further details see, for instance, Remark \ref{remaadaptive} below). 

\paragraph{Our Contribution.}
In order to provide results that hold for both influence maximization and stochastic probing (and for variants or generalizations of them), we will first introduce the framework of \emph{Stochastic Budgeted Multi-round Submodular Maximization} (SBMSm) problems. These problems are defined by the following quantities: a number $T$ of rounds; a ground set $V$ of $n$ items, whose state is drawn at each time step from an item-specific set according to some time-dependent distributions; a monotone submodular objective function $f_t$ at each round $t$, defined on subsets of $V$; a budget $B$ on the possible selectable items within the time-horizon $T$. The problem is to choose which items to select in each of the $T$ rounds, subject to the budget limit, in order to maximize the expected total value of the objective functions over the $T$ rounds. Furthermore, the items can be selected in an {\em adaptive} way, that is, after selecting an item, one observes the related state and the new value of function $f_t$, and based on this information one can decide whether to select other items in the current round, or whether to save the budget to select items in subsequent rounds.

This problem extends and generalizes both the influence maximization and stochastic probing problems (more details will be provided below). In the first case, an item is a node $v$ that can be selected as a seed, the state of an item $v$ defines which nodes would be influenced if $v$ is a seed, with the distribution of the state of $v$ depending on edges' probabilities, and the objective function counts the number or the total weight of influenced nodes. In the second case, the state of an item can be active or inactive, and we can observe the state of selected items only; the value of the submodular set function depends on the set of selected items which are active.
Note that SBMSm not only generalizes influence maximization and stochastic probing problems
to multiple round settings, but it also allows to model more realistic settings where states can change over time: e.g., in the influence maximization setting, we may allow edge probabilities or node weights to change in different rounds or be discounted over time.

We will start by characterizing when the problem can be optimally solved. In particular, we first define a certain optimization problem parametrized by a round $t\leq T$ and a budget $b\leq B$, called {\em Single-Round Maximization} ($\SRM$). Then we show that, whenever $\SRM$ can be optimally solved in polynomial time, there is a polynomial time dynamic programming algorithm that returns the optimal solution for the general multi-round problem SBMSm. Note that the optimal adaptive selection of items over a multiple-round window, as required by SBMSm, is assumed to be very strong: at each time step, based on observations both in the past and in the current round, the algorithm decides if to spend a further unit of budget in the current round to select a new item, or to stop the selection of items for the current round, and save budget for next rounds. Note that this will allow the optimum to react immediately to very negative events (e.g., the currently chosen seeds have influenced much less nodes than expected) or to very positive events (e.g., the currently chosen seeds have influenced much more nodes than expected). Hence, we find it is quite surprising that the optimality of the single-round problem straightly extends to multiple rounds.

Unfortunately, $\SRM$ turns out to usually be a hard problem to solve. Indeed, it can be easily reduced from the (deterministic) monotone submodular maximization problem, that is known to be APX-hard \cite{Feige98}, and there is no algorithm that approximates this problem in general instances within a factor better than $1-1/e$ (unless P=NP). Nevertheless, we show that there are non-trivial instances for which the problem can be effectively solved in polynomial time, essentially corresponding to the cases in which there are a few items and the possible states of the underlying stochastic environment are polynomial in the size of the input instance. The obtained findings imply that both $\SRM$ and SBMSm are {\em fixed-parameter tractable} (FTP) problems \cite{DowneyF99} w.r.t. to the number of items, if the outcomes of the stochastic environment are polynomially bounded. 

A natural direction would be to extend our framework so that it would work also with similar but more general instances where, for example, one may compute an arbitrarily good approximation in polynomial time. However, within our dynamic programming framework, this approximation ratio would explode with the number of rounds, and hence fails to provide good approximation guarantees.

For this reason, we need to resort to a different framework. Interestingly, we will show that a greedy polynomial time algorithm, that is arguably simpler than the previous one, provides a $1/2(1-1/e-\epsilon)\approx 0.316$-approximation of the optimal adaptive solution for every instance, where $\epsilon>0$ is an arbitrarily small constant due to the use of some randomized sampling procedures. Even more surprisingly, we can achieve such result in two phases: (i) we first preallocate the budget to rounds greedily, i.e., by assigning each unit of budget to the round in which it is expected to produce the largest expected value of the objective function, given the amount of budget still available; (ii) then, we apply the adaptive-greedy algorithm \cite{golovin2011adaptive} that chooses, at each round, the items to observe up to the allocated budget in a greedy adaptive way, i.e., it sequentially chooses an item that is expected to achieve the approximate maximum value of the objective function with high probability (with both approximation and probability depending on the choice of $\epsilon>0$), observes the realization of the state for this item, and then it chooses the next item among the ones that are expected to maximize the objective function given the state realization for the first item, and so on. To show the $1/2(1-1/e-\epsilon)$-approximation factor, we first prove that in a SBMSs problem there exists a non-adaptive allocation of budget to rounds that guarantees a $1/2$-approximation of the optimal adaptive solution. Then, we show that the adaptive greedy algorithm (considered in the phase (ii) described above) applied to such budget allocation further lowers the approximation guarantee by a factor of $(1-1/e-\epsilon)$. Finally, we show that the best non-adaptive allocation of budget for the adaptive greedy algorithm is that determined by preallocating the budget to rounds greedily, as in the phase (i) described above. 

In other words, our result shows that the powerful optimal adaptive solution for an SBMSs problem, that is adaptive both in the allocation of budget to rounds and in the choice of selected items within rounds, can be well-approximated by an algorithm that is only partially adaptive, i.e., it is adaptive in the choice of selected items, but not in the allocation of budget to rounds. 

Apart from the good approximation guarantees achieved by our greedy algorithm, the non-adaptive budget allocation is of independent interest, as it makes sense for several multi-round settings arising, for instance, from marketing or political campaigns. In fact, assuming that the considered campaign must be planned in multiple periods (i.e. rounds), it may be preferable to allocate the budget to be spent for each period in advance, due to organizational and management reasons. Therefore, for this and other settings, it would be good to measures how much an optimal policy for SBMSm, that is adaptive in both the budget allocation and the choice of the selected items, is much better than a simpler partially adaptive policy that, instead, determines the budget to be spent for each round in a non-adaptive way (that is, before starting the adaptive selection of the items and without observing their realizations).

To this aim, we introduce a metric called {\em budget-adaptivity gap} (see Definition \ref{def:bud_ad_gp} below), defined as the worst-case ratio between the fully  optimal adaptive value and that guaranteed by a partially adaptive policy that non-adaptively allocates the budget. The notion of budget-adaptivity gap can be seen as a novel and interesting variant of the standard adaptivity gap \cite{dean2008approximating,dean2005adaptivity}, that compares optimal adaptive strategies with fully non-adaptive ones, mainly for single-round optimization problems. Indeed, our variant still estimates the impact of non-adaptivity, but on the budget allocation only, that is a fundamental aspect to be taken into account to adaptively maximize the objective functions in multi-round settings.

We show that the budget-adaptivity gap is between $e/(e-1)\approx 1.582$ and $2$, that is, the non-adaptive budget allocation, in general, cannot be used to reach the fully adaptive optimum but still constitutes a good approximation of the fully adaptive ones; this result also implies that the above greedy algorithm essentially guarantees the best approximation among all partially adaptive policies. We point out that the effectiveness of a partially adaptive approach for multi-round adaptive optimization has been left as prominent research direction by \cite{auletta2022augmented}, and to the best of our knowledge, our results are the first that consider such an approach and guarantee, at the same time, a good approximation to the optimal solution in polynomial time.


\section{Model and Definitions}
\label{sec:model}

Before describing the framework of Stochastic Budgeted Multi-round Submodular Maximization, we introduce some preliminary notation, a general class of maximization problems over multiple rounds, and the main tools of adaptive optimization.

\paragraph{Preliminary Notation.}
Given two integers $k_1,k_2$, let $\llbracket k_1,k_2\rrbracket:=\{k_1,\ldots, k_2\}$ be the discrete interval from $k_1$ to $k_2$ (that is empty if $k_1>k_2$); furthermore, let $[k]$ denote the discrete interval $\llbracket 1,k\rrbracket$.

\subsection{A Setting of Multi-round Maximization}
We consider a setting $I:=(T, B, V, H, (\mathcal{P}_t)_{t\in [T]},(f_t)_{t\in [T]})$ whose parameters and connected properties are described as follows.

$T\geq 2$ is a finite and integral {\em horizon}, where each $t\in [T]$ denotes the $t$-th {\em round}. $B \geq 0$
is a given {\em budget}.
$V$ is a {\em ground set} of $n\geq 2$ {\em items}. For any fixed $v\in V$, let $H(v)$ be a set of possible {\em local states} for item $v$, and let $H:=\times_{v\in V}H(v)$ be the set of {\em global states}, that is, a global state $\eta\in H$ is a collection of local states $(\eta(v))_{v\in V}$ (with $\eta(v)\in H(v)$ for any $v\in V$). In the following, the global states may be simply called states. 

For any round $t$, a state $\eta_t\in H$ is picked according to a probability distribution $\mathcal{P}_t$ over $H$, that is called {\em (global) state distribution} at round $t$. The state distributions $\mathcal{P}_1,\ldots, \mathcal{P}_T$ are assumed to be independent; however, for any fixed round $t\in [T]$, the probability distributions induced by local states $\eta_t(v)$ from $\mathcal{P}_t$ may be dependent. We observe that, even if we assumed that each tuple  $\eta:=(\eta(v))$ with $\eta(v)\in H(v)$ is a global state, we can impose the unfeasibility of $\eta$ at some round $t$ by setting $\mathcal{P}_t(\eta)=0$ (i.e., $\eta$ never occurs at round $t$). 
In the following, we will adopt bold letters to denote the random variables, and non-bold letters to denote their realization (e.g., $\bm \eta(v)$ is a random local-state, and $\eta(v)$ denotes a possible realization of $\bm \eta(v)$).

Finally, for any round $t\in [T]$, $f_t:2^V\times H \rightarrow \RP$ is the {\em objective function at round $t$}, that associates a non-negative real-value $f_t(S_t,\eta_t)$ with each subset of items $S_t\subseteq V$ and state $\eta_t$.
Note that objective functions are allowed to be different at each time step.
Such a variability in both the state distribution and objective function, immediately makes inefficient to allocate the same budget to each round: indeed, it is immediate to build a setting in which many rounds may lead to very poor outcomes, and only few of them may lead to a very large outcome, and hence all the budget allocated to the former rounds is wasted (see Remark \ref{remalimit} for further details). 

Let $\fT:(2^V\times H)^T\rightarrow \RP$ be the {\em aggregated objective function}, that assigns the value $\fT(\vec{S},\vec{\eta}):=\sum_{t=1}^T f_t(S_t,\eta_t)$ to each {\em aggregated subset} $\vec{S}:=(S_1,\ldots, S_T)$ and {\em aggregated state} $\vec{\eta}:=( \eta_1,\ldots,  \eta_T)$, i.e., it returns the cumulative value of the objective function $f_t$ over all rounds $t\in [T]$.

We consider then the problem of selecting an aggregated subset $\vec{ S}=(S_1,S_2,\ldots, S_T)$ such that $\fT(\vec{S},\vec{\eta})$ is maximized subject to the budget constraint $\sum_1^T|S_t|\leq B$, for a fixed aggregated state $\vec{ \eta}$ drawn from the aggregated state distributions $(\bm \eta_1,\ldots, \bm \eta_T)$. Unfortunately, we need to solve this problem without knowing the aggregated state, and we have access to the state distributions and objective functions only.

\subsection{Maximization via Adaptive Policies}
We here focus on algorithms for selecting the aggregated subset $\vec{ S}$ achieving (an approximation of) the optimal value of $\fT(\vec{S},\vec{\eta})$ that are \emph{adaptive}, i.e., they select items over subsequent rounds, observe the partial states  of the probed items and use these observations to guide the future choices in an adaptive fashion.

We next provide a formal definition for this adaptive behavior. Let us first introduce some further notation.
For a given state $\eta_t$ at round $t$, we let $\langle S_t,\eta_t\rangle = (\eta_t(v))_{v \in S_t}$ be the \emph{partial state} observed when selecting the items in $S_t$ at round $t$.
As the state $\eta_t$ is generally unknown, we will simply denote $\langle S_t, \eta_t \rangle$ as $\langle S_t\rangle$.
An {\em aggregated partial state} is a vector $\langle \vec{S}\rangle:=(\langle S_1\rangle,\ldots, \langle S_T\rangle)$ of partial states, one at each round, and collects the observations done at each round.

We look for a {\em multi-round (adaptive) policy} $\pi$, that on input the (unknown) aggregated state $\vec{\eta}$ returns an aggregated subset $\vec{S}:=\pi(\vec{\eta})$, according to the following procedure:\newline
\begin{enumerate}
\item It initially starts from the first round and an empty aggregated set, and it initializes the adaptive strategy by setting
$t=1$ (the actual round), $\vec{S}:=(S_1,\ldots, S_T):=(\emptyset,\ldots, \emptyset)$ (the actual aggregated set), and $i=1$ (the actual \emph{step} at round $t$)\footnote{We highlight that we will use the term \emph{round} to define the $T$ phases over which our optimization problem is defined. In each round, the adaptive policy can be in turn  executed over multiple phases: these will be termed \emph{steps}. Finally, we use the term \emph{time} only with reference to computational time.}.
\item At each round $t$ and step $i$, either it (a) selects an item $v_{t,i}\in V$ not included in $S_t$ yet, or (b) it decides to stop and proceed to next round. \label{item:step2}
\begin{itemize}
\item[(a)] If it selects an item $v_{t,i}$, it adds it to $S_t$ and sets $i\leftarrow i+1$, i.e., the policy is ready for the next item selection in the same round.
Having selected the item $v_{t,i}$, it can observe the local state $\eta_t(v_{t, i})$.
Hence, the policy knows the partial state $\langle S_t,\eta_t\rangle$ and it can use such knowledge to select the (eventual) next item to be probed and included in $S_t$.
\item[(b)] If the policy decides to cease in selecting items in the current round, then it sets $t\leftarrow t+1$ and $i\leftarrow 1$.
\end{itemize}
\item Repeat step~\ref{item:step2} until $t>T$ (i.e., the horizon has been exceeded) or $\sum_{t'=1}^t|S_{t'}|=B$ (i.e., all the budget has been used).
\end{enumerate}

Among the possible multi-round policies, we are interested in the ones that return (an approximation of) the maximum $\sigma(\pi):=\mathbb{E}_{\vec{\bm \eta}\sim \mathcal{P}^T}\left[\fT(\pi(\vec{\bm \eta}),\vec{\bm \eta})\right]$, (where $\pi(\vec{\bm \eta})$ is the aggregated set returned by $\pi$ and $\mathcal{P}^T$ denotes the product distribution $\times_{t\in [T]}\mathcal{P}_t$). Note that the optimal value achievable by a multi-round adaptive policy can be not so high as that obtained if we knew the input aggregated state, but it is the best we can hope by inferring the aggregated state via the adaptive probing of the items, based on the knowledge of the state distribution.

An interesting class of the above policies is that of {\em partially adaptive policies}, in which the number of items that must be selected at each round (i.e., the budget per round) is determined non-adaptively in advance, and then the adaptive selection of items is performed within each round independently. Partially adaptive policies are computationally simpler than general adaptive ones, but it might have worse performance guarantees. Actually, we will show that this worsening is limited for our problem.


Next,
with a little abuse of notation, we
will often
use
the notation
$\pi(I,a_1,\ldots)$ to denote a policy applied to $I$, where $a_1,\ldots$ are further parameters defining the policy, and the (unknown) aggregated state $\vec{\eta}$ will be included as further input parameter.

\subsection{Properties of the Objective Functions}\label{subsec: properties}
In the following, we list some additional properties that the objective functions must verify in our setting.

The first subset of properties will focus on the complexity of the objective function.
Specifically, we assume that each objective function $f_t$ is representable and computable in polynomial time.
Moreover, we also constraint the maximum value of an objective function, as follows.
\begin{definition}
An objective function  $f_t$ is {\em polynomially bounded} if there exist real values $\lambda,\Lambda>0$, with $\Lambda$ and $1/\lambda$ being polynomially bounded in the size of $I$, such that the following conditions hold: (i) $f_t(\{v\}\cup S_t,\eta_t)-f_t(S_t,\eta_t)\leq \Lambda$ for any $v\in V$, $S_t\subseteq V$ and state $\eta_t$; (ii) there exist $t\in [T]$ and $v\in V$ such that $\lambda\leq \mathbb{E}_{\bm \eta_t\sim \mathcal{P}_t}[f_t(\{v\},\bm \eta_t)]$.
\end{definition}
The second subset of properties instead extends the typical feature of monotonicity and submodularity to our multi-round adaptive setting.
Specifically, we assume that, for any fixed state $\eta_t$, the function $f_t$ satisfies the following {\em monotonicity} property: $f_t(S_t,\eta_t)\leq f_t(S_t',\eta_t)$ if $S_t \subseteq S_t'$.


We further assume that objective functions satisfy a certain property connected with the notion of submodularity. Given two partial states $\langle S_t\rangle:=\langle S_t,\eta_t\rangle$ and $\langle S'_t\rangle:=\langle S'_t,\eta'_t\rangle$, we say that $\langle S_t\rangle$ is a {\em sub-state} of $\langle S'_t\rangle$ (equivalently, we say $\langle S_t\rangle\preceq \langle S'_t\rangle$)  if $S_t\subseteq S'_t$,
and $\eta_t(v)=\eta'_t(v)$ for any $v\in S_t$, 
that is, the local states observed in $\langle S_t\rangle$ coincide with those observed in the restriction of $\langle S'_t\rangle$ to set $S_t$; we observe that $\langle V,\eta_t\rangle=\eta_t$ holds for any state $\eta_t$ (by definition).
Given a partial state $\langle S_t\rangle$ at round $t$ and $v\in V$, let $
\Delta_t(v \mid \langle S_t\rangle):=\mathbb{E}_{\bm \eta_t\sim\mathcal{P}_t}\left[f_t(S_t\cup \{v\},\bm \eta_t)-f(S_t,\bm \eta_t)|\langle S_t\rangle\preceq\bm \eta_t\right]$ i.e., the expected increment of
$f_t$ obtained when adding item $v$ to $S_t$, conditioned by the observation of partial state $\langle S_t\rangle$.

We then have the following definition, that adapts the notion of adaptive submodularity introduced by \cite{golovin2011adaptive} to our setting.
\begin{definition}[\cite{golovin2011adaptive}]\label{def:ad_subm}
An objective function $f_t:2^V\rightarrow \RP$ is {\em adaptive submodular} if, for all $\langle S_t\rangle\preceq \langle S'_t\rangle$ and $v\in V$ we get $\Delta_t(v \mid \langle S_t\rangle)\geq \Delta_t(v \mid \langle S'_t\rangle)$.
\end{definition}
In this work, we strengthen the adaptive submodularity by introducing the {\em strong adaptive submodularity}, that is defined as follows.
\begin{definition}\label{def:strong_ad_subm}
An objective function $f_t:2^V\rightarrow \RP$ is {\em strongly adaptive submodular} if the following properties hold: (i) It is adaptive submodular. (ii) Given an integer $k\geq 1 $, let $\pi$ be an optimal adaptive policy that selects $k$ items. Then, for any partial state $\langle S_t\rangle$ and any adaptive selection $\tilde{\bm S}_t\supseteq S_t$ of at most $k+|S_t|$ items conditioned by the observation of $\langle S_t\rangle$, we have
\begin{equation}
\sigma(\pi)\geq \mathbb{E}_{\bm \eta_t\sim \mathcal{P}_t}[f_t(\tilde{\bm S}_t,\bm \eta_t)-f_t(S_t,\bm \eta_t)|\langle S_t\rangle\prec \bm \eta_t].
\end{equation}
\end{definition}

Intuitively, adaptive submodularity (i.e., property (i) of strong adaptive submodularity) requires that the marginal increment of the expected value of \( f_t \) due to the selection of the same item \( v \) does not increase as the observations become richer. In contrast, the second property of strong adaptive submodularity requires that the expected increment of \( f_t \), achieved by the adaptive selection of \( k \) items after observing a given partial state, is not better than the optimal adaptive value for selecting \( k \) items, without any prior item selection and observation.
%

From this point forward, we assume that the objective functions satisfy all the properties described above.
\subsection{The Stochastic Budgeted Multi-round Submodular Maximization Problem}
We denote the setting described above as {\em Stochastic Budgeted Multi-round Submodular Maximization Problem} (SBMSm), that looks for a multi-round policy $\pi$ maximizing $\sigma(\pi)=\mathbb{E}_{\vec{\bm \eta}\sim \mathcal{P}^T}\left[\fT(\pi(\vec{\bm \eta}),\vec{\bm \eta})\right]$
on an input instance $I=(T,B,V,H,(\mathcal{P}_t)_{t\in [T]},(f_t)_{t\in [T]})$.

In particular, we will focus on policies with guaranteed approximation,
defined as follows.
Given an instance $I=(T,B,V,H,(\mathcal{P}_t)_{t\in [T]},(f_t)_{t\in [T]})$ of SBMSm, $\alpha\in [0,1]$,
and a multi-round policy $\pi$ for $I$, we say that $\pi$ is \emph{$\alpha$-approximate} (or, equivalently, $\alpha$ is the {\em approximation factor} of $\pi$), if $\sigma(\pi)\geq \alpha\cdot \sigma(\pi^*)$, where $\pi^*$ denotes an optimal multi-round policy for $I$; 


We observe that the size of the input instance $I$ generally depends on the succinct representation of all the parameters defining $I$. For the sake of simplicity, in the remainder of the paper we will implicitly assume that the size of the input instance is polynomial in $n$ and $T$. The polynomial dependence in the number of items $n$ is a standard assumption/property of problems connected with the optimization of functions defined on subset of items. The polynomial dependence in $T$ holds since the input instance $I$ can be seen as a collection of $T$ single-round sub-instances (one for each round $t\in [T]$), which are assumed to be explicitly listed. Furthermore, we will implicitly assume that $B<n T$, where $B$ is the budget of the input instance. Indeed, if $B\geq nT$, we could select a set $V$ at each round so that the objective function is trivially maximized at each round. Thus, the assumption $B<nT$ implies that the budget $B$ is also polynomial in the size of $I$.
Finally, we allow the size of the remaining parameters defining $I$ (i.e., local and global states, probability distributions and objective functions) to be exponential when explicitly represented, although the size of their succinct representation is not; this is the case, for instance, of influence maximization problems, which are described below. 

Compatibly with several stochastic optimization problems, we also assume the existence of a polynomial-time oracle that works as follows: (i) for each $t\in [T]$, $S\subseteq V$ and state $\eta\in H$ given in input, it is able to compute $f_t(S,\eta)$; (ii) for each $t\in [T]$, sets $S,S'$ with $S\subseteq S'\subseteq V$ and partial state $\langle S\rangle$, the oracle is able to pick a partial state $\langle S'\rangle\succeq \langle S\rangle$, drawn from the distribution $\mathcal{P}_t$ conditioned by the observation of $\langle S\rangle$.

Before providing our results, we will provide a description on how to apply our framework to multi-round and adaptive variants of influence maximization and stochastic probing.

\subsection{Applications}\label{sec:appl}
\subsubsection{Influence Maximization}
The Influence Maximization problem has been extensively studied in the last years under different models of information diffusion. Here we focus on the Independent Cascade diffusion model but a similar approach can be used to model the Influence Maximization under different diffusion models. Before to show how this problem can be fitted in the class of SBMSm, let us first define the problem of Influence Maximization under the Information Cascade model and introduce its adaptive and multi-round variants.

\paragraph{Definition.}
The \emph{Independent Cascade diffusion model} (IC) \cite{kempe2003maximizing}, given a graph $G = (V, E)$, a probability $p \colon E \rightarrow [0, 1]$ defined on edges and a set of \emph{seeds} $S_0 \subseteq V$, returns a set $A=IC(G,p,S_0)$ of \emph{active nodes} iteratively computed over multiple (discrete) diffusion rounds\footnote{Here, ``diffusion round" has a different meaning from the term ``round" considered in SBMSm.} as follows: (i) at diffusion round $0$, the set of activated nodes $A_0$ coincides with $S_0$; (ii) given a diffusion round $h\geq 1$, the set of nodes $A_h$ activated at time step $h$ is determined by including in $A_h$, for each edge $(u,v)$ with $u\in A_{h-1}$, the node $v$ with probability $p_{(u,v)}$, independently of the other edges; (iii) when the diffusion process reaches a diffusion round $k$ with $A_k=\emptyset$, that is, no further node can be activated, the diffusion process ends and the overall set of active nodes is $A:=\bigcup_{h=0}^{k-1} A_h$. We observe that, in the IC model, each active node has only one chance to activate each of its neighbors, the activation process associated with each edge is independent from the others.

In the problem of \emph{Influence Maximization under Independent Cascade diffusion model} (IM-IC) \cite{kempe2003maximizing} we are given a graph $G = (V, E)$,
a weight function $w \colon V \rightarrow \mathbb{R}_{\geq 0}$ defined on nodes,
a probability $p \colon E \rightarrow [0, 1]$ defined on edges, and a budget $B$, and we are asked to select the set $S^*$ of at most $B$ nodes that maximizes the expected sum of weights of active nodes returned by the IC diffusion model on input $G$ and $p$ i.e., $S^* = \arg \max_{S \subseteq V \colon |S| \leq B} \mathbb{E}_p \left[
\sum_{v \in IC(G, p, S)} w(v)
\right]$.

Recently, an adaptive variant of the problem has been proposed where seeds are introduced one at a time, using an adaptive policy. To define the adaptive variant of the problem we need  first to define the concept of adaptive policy. For any graph $G = (V, E)$ and set $U\subseteq V$ of nodes, let $G[U]$ denote the subgraph of $G$ induced by the nodes in $U$. Given a graph $G = (V, E)$, a probability $p \colon E \rightarrow [0, 1]$ defined on edges, and a budget $B$, an \emph{adaptive policy} $\pi$ for IM-IC works as follows: it selects a first vertex $v_1$, observes the set $A_1$ of nodes activated by seed $v$, i.e., $A_1 = IC(G, p, \{v_1\})$, then selects a new vertex $v_2 \in V \setminus A_1$, observes the set $A_2$ of newly activated nodes, i.e., $A_2 = IC(G[V\setminus A_1], p, \{v_2\})$, then selects a new node $v_3 \in V \setminus (A_1 \cup A_2)$, observes the set $A_3$ of newly activated nodes, i.e., $A_3 = IC(G[V\setminus (A_1\cup A_2)], p, \{v_3\})$, and so on until either there is no node that can be selected or $B$ nodes have been already selected. The set of nodes activated through an adaptive policy $\pi$ is $A(\pi) = \bigcup_{i \in [B]} A_i$.

In the problem of \emph{Adaptive Influence Maximization under Independent Cascade diffusion model} \cite{golovin2011adaptive} we are given a graph $G = (V, E)$,
a weight function $w \colon V \rightarrow \mathbb{R}_+$ defined on nodes,
a probability $p \colon E \rightarrow [0, 1]$ defined on edges, and a budget $B$, and we are asked to select the adaptive policy $\pi^*$ that maximizes the expected sum of weights of activated nodes, i.e., $\pi^* = \arg \max_{\pi} \mathbb{E}_p \left[
\sum_{v \in A(\pi)} w(v)
\right]$.

In this paper we are interested to multi-round settings. Hence, we need to extend previous concepts to the multi-round setting.
Given a graph $G = (V, E)$, a probability $p_t \colon E \rightarrow [0, 1]$ defined on edges (for each round $t$), a budget $B$, and a time horizon $T$, a \emph{multi-round adaptive policy} $\pi$ for IM-IC works as follows: at each round $t \in [T]$, let $S^t(\pi)$ be the set of seeds selected by $\pi$ in round $t$ and $A^t_v(\pi)$ be the set of nodes activated by $v \in S^t(\pi)$ (we will next omit the reference to $\pi$ when clear from the context). Note that both $S^t$ and $A^t_v$ are initially empty. Then, the policy at each round $t \in [T]$, if the total number of selected nodes is lower than $B$, works as follows: 
\begin{enumerate}
 \item It decides whether to choose another seed in the current round or to move to the next round (note that if all nodes have been activated in the current round, then the policy must necessarily choose to move to the next round). \label{it:step1}
 \item In the first case, the policy selects a node $v$ that has not been yet activated in the current round, i.e., $v \notin A^t_{S^t} := \bigcup_{u \in S^t} A^t_u$, observes the set $A^t_v$ of newly activated nodes, i.e., $A^t_v := IC(G[V\setminus A^t_{S^t}], p, \{v\})$, adds $v$ to $S^t$, and repeat from step~\ref{it:step1}. 
\end{enumerate}

Finally, in the problem of \emph{Multi-Round Adaptive Influence Maximization under Independent Cascade diffusion model} (MR-IM-IC) we are given a graph $G = (V, E)$, a time horizon $T$,
a weight function $w_t \colon V \rightarrow \mathbb{R}_+$ defined on nodes for every $t \in [T]$,
a probability $p_t \colon E \rightarrow [0, 1]$ defined on edges for each round $t\in [T]$, and a budget $B$, and we are asked to select the multi-round adaptive policy $\pi^*$ that maximizes the expected sum of weights of activated nodes over $t$ rounds, i.e., $\pi^* = \arg \max_{\pi} \varsigma(\pi)$ where $\varsigma(\pi) = \mathbb{E}_{(p_t)_{t\in [T]}} \left[\sum_{t \in [T]}
\sum_{v \in A^t_{S^t}(\pi)} w_t(v)
\right]$.

We stress that our definition allows that both nodes' weights and edge probabilities may change from one round to the next: this can be useful to implement a variable discount factor (e.g., we care more about the number of people influenced in the last phases of the campaign) or to implement a dynamic environment in which the members of a population are known to appear and disappear over time, or the relevance of population members changes over time (e.g., in liquid democracy settings \cite{blum2016liquid} at the beginning we are more interested in influencing a large portion of the population for manipulating the choice of delegates, while later in the campaign we are more concerned in influencing those members that have more chances to be chosen as delegates, in order to manipulate the outcome of their voting process).

In Appendix \ref{apppendix:applications} we show that MR-IM-IC can be modeled as an SBMSm problem satisfying all required properties.

\subsubsection{Stochastic Probing}
\paragraph{Definition.} In general \emph{stochastic probing} problems, we have $n$ items, each one associated with a random state independent from the others, and the goal is to choose a feasible subset of items maximizing the value of a certain function depending on the states of the selected items. Several classes of stochastic probing problems arise from particular choices of the states and the function to be maximized.

The stochastic probing problem we consider here, denoted as {\em binary submodular stochastic probing} (BSS), is defined as follows: We have a set $V$ of $n$ items, and a non-negative, monotone and submodular set function $g:2^V\rightarrow \RP$ that we aim at maximizing. Each item $v$ has a binary state $\phi_v$ in the set $\{active, inactive\}$, the probability that $v$ is a active is $p(v)$, and the distributions associated to all items are independent. Given a subset $S\subseteq V$ and a realization $\phi=(\phi(v))_{v\in V}$ of the states, let $act(S,\phi)\subseteq S$ denote the subset of active items in $S$. For a given budget $B\geq 0$, the goal is to find a subset $S^*\subseteq V$ of at most $B$ items that maximizes $\mathbb{E}_{\phi}[g(act(S^*,\phi))]$ (i.e., the expected value of $g$ applied to the selected items which are active).

The multi-round adaptive variant of this problem, denoted as {\em multi-round binary submodular stochastic probing (MR-BSS)}, assumes that the probabilities and the submodular set functions are time/round specific (i.e., we have a probability distribution $p_t$ and a function $g_t$ for any round $t$), the budget (i.e., the number of items we can select) can be adaptively allocated over all possible rounds, and the goal is to maximize the sum of the values achieved by all functions $g_t$. 
In the adaptive variant of the problem, one considers adaptive multi-round policies that adaptively select the items at each round. Furthermore, we assume that there are two quantities $\lambda,\Lambda>0$ that are polynomial in the input instance such that: (i) the increment caused to $g_t$, when probing a new item $v$, is upper bounded by $\Lambda$; (ii) there exists one round $t$ and an item $v$ such that the expected value of $g_t$ when probing item $v$ at round $t$ is at least $\lambda$.

In Appendix \ref{apppendix:applications} we show that MR-BSS can be modeled as an SBMSm problem satisfying all required properties.


\section{Optimal Algorithm}
In this section we first describe a single-round optimization problem of interest, and how its solution is used to provide a dynamic programming algorithm that will return the optimal policy in polynomial time as long as the single-round optimization problem can be efficiently addressed. Then, we will discuss a few of cases in which the single-round optimization problem can be efficiently addressed.

These results will imply that, if the number of states $|H|$ is polynomially bounded\footnote{In the model definition, we already observed that $|H|$ might be exponential (e.g., for influence maximization problems).} in the size $|I|$ of the input instance, SBMSm is {\em fixed-parameter tractable} (FTP) \cite{DowneyF99} w.r.t. the number of items $n$, that is, the complexity of the (optimal) dynamic programming algorithm is $O(F(n)\cdot |I|^{p})$ for some computable function $F$ and $p>0$.

Finally, we will highlight the limits of this dynamic programming framework, by discussing when it can be extended to provide a good approximation of the optimal policy, even if the single-round optimization problem can only be approximated.

We stress that the properties given in Section \ref{subsec: properties} are not required to hold in the results provided in this section (differently from the subsequent sections).

\subsection{A Dynamic Programming Algorithm based on Single-round Optimization}
Let us fix a round $t$ and a budget $b$ available. Given $i\in \llbracket 0,\min\{n,b\}\rrbracket$, let $R(t+1,i)$ denote the optimal expected value of $\sum_{t'=t+1}^T f_{t'}$ achievable in rounds $t+1,\ldots, T$ by a multi-round policy $\pi_{\geq t+1}$ that is restricted to such rounds only and uses at most $i$ units of budget; we observe that $R(t+1,i)=0$ for $t=T$. We assume that $R(t+1,i)$ has been already determined for any $i\in \llbracket 0,b\rrbracket$. The {\em single-round maximization problem} at round $t$ and budget $b$ ($\SRM$) asks to find a policy $\pi_t$ at round $t$ that adaptively computes the budget $i(\pi_t,\bm \eta_t)\in \llbracket 0,b\rrbracket$ to be allocated at round $t$ in such a way that the expected value 
\begin{equation}\label{optimum_singleround}
\mathbb{E}_{\bm \eta_t\sim \mathcal{P}_t}[f_t(\pi_t(\bm \eta_t),\bm \eta_t)+R(t+1,b-i(\pi_t,\bm \eta_t))]
\end{equation}
of objective function $f_t$ plus $R(t+1,b-i(\pi_t,\bm \eta_t))$ is maximized, where the random budget $i(\pi_t,\bm \eta_t)$ used at round $t$ depends on the chosen policy $\pi_t$ and the state $\bm \eta_t$. 

We observe that the optimal value of \eqref{optimum_singleround} coincides with $R(t,b)$, that is, the optimal expected value of $\sum_{t'=t}^T f_{t'}$ achievable in rounds $t,\ldots, T$. Thus, the following recursive relation holds:
\begin{equation}\label{recurrence}
R(t,b)=
\begin{cases}
0,&\text{ if }t=T+1\\
\max_{\pi_t}\mathbb{E}_{\bm \eta_t\sim \mathcal{P}_t}[f_t(\pi_t(\bm \eta_t),\bm \eta_t)+R(t+1,b-i(\pi_t,\bm \eta_t))],& \text{ if }t\in \llbracket 0,T\rrbracket.
\end{cases}
\end{equation}
Suppose that there is an algorithm for $\SRM$, that is, by resorting to the recursive relation \eqref{recurrence}, we can efficiently compute $R(t,b)$ as a function of $R(t+1,i)$ for any $i\in \llbracket 0,b\rrbracket$, together with the optimal single-round policy. Such algorithm can be used as a sub-routine for a more general dynamic programming algorithm that computes an optimal policy for the general SBMSm problem. 

Specifically, the dynamic programming algorithm first sets $R(T+1, b) = 0$ for every $b \in \llbracket 0, B\rrbracket$. By resorting to the recursive relation \eqref{recurrence}, we can compute $R(T, b)$ for any $b \in \llbracket 0, B\rrbracket$, as a function of the values $R(T+1,b-i)$ with $i\in \llbracket 0,\min\{n,b\}\rrbracket$. Iteratively, going from $t=T-1$ to $t=1$, we can compute $R(t, b)$ for any $b \in \llbracket 0, B\rrbracket$, as a function of the values $R(t+1,b-i)$ with $i\in \llbracket 0,\min\{n,b\}\rrbracket$. Finally, we observe that $R(1,B)$ is the optimal value for the SBMSm problem. Thus, starting from round $t=1$ and budget $B$, we can reconstruct the optimal multi-round policy achieving the optimal value (this can be done by resorting to standard dynamic programming techniques).

\subsection{How to Solve Single-round Optimization Efficiently?}
In the above dynamic-programming algorithm, we assumed the existence of a sub-routine that solves $\SRM$ for any round $t$ and budget $b$, i.e., that allows to compute recurrence \eqref{recurrence}. Unfortunately, the $\SRM$ problem can be easily reduced from the submodular monotone optimization problem, that is NP-hard \cite{Feige98}.
To address the intractability of the single-round optimization problem, in Appendix \ref{sec_fpt} we describe a tree-based algorithm that, while generally inefficient, can be run in polynomial time in specific cases. In particular, we show that this algorithm is FPT with respect to $n$, provided the number of states $|H|$ is polynomially bounded in the input instance. Consequently, by incorporating this algorithm as a subroutine in the dynamic programming approach described above, we obtain an FPT algorithm for the more general SBMSm problem, as outlined in the following theorem.
\begin{theorem}[Corollary \ref{corollary_fpt} of Appendix \ref{sec_fpt}]
Both $\SRM$ and SBMSm are FPT problems w.r.t. to $n$, if $|H|$ is polynomially bounded in the size of the input instance.
\end{theorem}

%

While the single-round optimization problem is known to be NP-hard, it would be still possible to efficiently find arbitrarily good approximation of their optimal solution in many settings. Hence, it would be natural to replace in our dynamic programming framework the optimal solution to the single-round optimization problem with its approximation. However, it is not hard to observe that this causes an explosion on the approximation guarantee: indeed, an $\alpha$-approximation of the optimal policy at round $t$ evaluated on a game tree in which the values of $R(t+1, b)$ for the different choices of $b$ are only $\beta$-approximations of the maximum expected revenue achievable in round $t+1, \ldots, T$ with budget $b$, may be not able to return a policy guaranteeing an approximation better than $\alpha \cdot \beta$ of the the maximum expected revenue achievable in round $t+1, \ldots, T$ with the given budget. Hence, by running a $\rho$-approximation algorithm for all $T$ steps, may only allows to compute a $\rho^T$-approximation of the optimal multi-round adaptive policy, that is clearly an unsatisfying approximation guarantee for every constant $\rho$.

\section{Partially Adaptive Approximation Algorithm}
\label{sec:opt}
In this section, we will design and analyze a partially adaptive policy as $1/2(1-1/e-\epsilon)\approx 0.316$-approximation algorithm for SBMSm, where $e\approx 2.718$ denotes the Nepero number and $\epsilon>0$ is an arbitrarily small error due to some sampling procedures. This policy is based on the greedy approach, and its approximation guarantee and time-complexity hold under the properties given in Section \ref{subsec: properties}. 

In general, a {\em partially adaptive policy}, in a first part, 
determines how to
distribute the budget $B$ over all rounds by computing a {\em budget vector} $\vec{b}=(b_1,\ldots, b_T)$, so that $b_t$ is the budget assigned to round $t$; this budget assignment is done in a non-adaptive way, that is, without knowing the particular realization of the stochastic environment, but only based on the knowledge of the input instance $I$. Subsequently, for each round $t\in [T]$, the policy adaptively selects $b_t$ items to be observed in round $t$; differently from the the first part, now the realization of the stochastic environment is partially revealed after each observation. It is worth to notice that an optimal solution for SBMSm, in general,
dynamically allocates the budget after each item selection.
Instead a partially adaptive policy allocates all the budget to the rounds at the beginning. 

The specific partially adaptive policy we will consider is based on the greedy approach: (i) it first determines the budget allocation by repeatedly assigning each unit of budget to the rounds in a non-adaptive greedy way, and then (ii) greedily selects the items in an adaptive way (according to the budget previously allocated to each round). 

Despite partially adaptive policies are in general weaker than fully adaptive ones (because of the relaxation on budget adaptivity), we will show that our policy still guarantees a good approximation of the optimal adaptive value. Indeed, as a consequence of the hardness result of \cite{Feige98} (holding for non-adaptive single-round submodular maximization), we have that SBMSm cannot be approximated within a factor better than $0.632\approx 1-1/e$ (unless P=NP). Thus, the approximation guaranteed by our partially adaptive policy is close to the known inapproximability bound.


\begin{remark}\label{remalimit}
Despite a partially adaptive policy decides the budget allocation in advance, it must strongly takes into account the heterogeneity of the input instance over all rounds, in order to obtain a good approximation of the optimal adaptive value. Indeed, if the budget was allocated without considering the structure of the input instance (e.g., by distributing uniformly the budget), the resulting approximation could be non-constant. 

To show this, assume that the budget allocation is determined by a budget vector $\vec{b}=(b_1,\ldots, b_T)$ that depends on the number of rounds $T$ and the budget $B$ only, and let $t^*$ be the index $t$ minimizing $b_t$ (among $b_1,b_2,\ldots, b_T$). Now, consider an instance $I= (T, B, V, H, (\mathcal{P}_t)_{t\in [T]}, (f_t)_{t \in [T]})$ where $|V|=B=T$ and $|H|=1$. As $|H|=1$, the underlying problem is deterministic, thus we can avoid at all the presence of states. Each $f_t$ is defined as follows: $f_t(S)=|S|$ for any $S\subseteq V$ if $t=t^*$ (that is, the objective function at round $t^*$ counts the number of selected items) and $f_t(S)=0$ if $t\neq t^*$ (that is, the objective function at any round $t\neq t^*$ is zeroed).

We have that the optimal policy $\pi^*$ for $I$ simply selects all items at round $t^*$ (that is, the unique round that produces a positive value of the aggregated objective function) and does not select any item at rounds $t\neq t^*$; the resulting value of $\pi^*$ is $\sigma(\pi^*)=B=T$. Instead, the value guaranteed by any policy $\pi$ that allocates the budget according to vector $\vec{b}$ is $\sigma(\pi)=b_{t^*}=\min_{t\in [T]}b_t\leq B/T=T/T=1$. Thus, policy $\pi$ is only $\alpha$-approximate with $\alpha=b_{t^*}/B \leq 1/T$ (that is non-constant due to dependence on $T$).
\end{remark}

Below, we provide the detailed description of such policy and the analysis of its approximation factor. 


\subsection{Estimating Objective Functions}
\label{sec:sampling}
In order to solve the desired maximization problems through our greedy partially adaptive policy, we first need to estimation the expected value of the objective functions achieved when selecting an item.
However, this is not an easy task due to the stochastic nature of the problem.
We first introduce a standard Monte-Carlo algorithm, called ${\sf Oracle1}$,
that can be
used to give a good estimation, with high probability, of the expected value of the objective function when selecting an item:
such algorithm can be used to find an item that approximately maximizes the expected influence spread of an objective function, under a given partial state.

Then, we provide another Monte-Carlo algorithm, called ${\sf Oracle2}$,
that can be
used to estimate the expected increment of the objective function achieved by a given adaptive policy $\pi$ when selecting the $j$-th item (in order of selection) at any round $t$; as in {\sf Oracle1}, the precision of such estimation will be good with high probability. By exploiting the Chernoff-Hoeffding's inequality \cite{hoeffding1963probability}, we show that, for any $\delta,\xi>0$ both algorithms can be set to guarantee a precision with an additive error of at most $\delta$ with probability at least $1-\xi$, and to run in polynomial time w.r.t. to the size of the input instance, $1/\delta$ and $1/\xi$.

\subsubsection{Oracle1}
\label{def_oracle}
${\sf Oracle1}$ takes in input an integer $q\geq 1$, an instance $I$ of the SBMSm problem, a round $t\in [T]$, a partial state $\langle S_t\rangle$ for round $t$, an item $v$, and picks $q$  samples $X_{t,v,\langle S_t\rangle}^1,\ldots, X_{t,v,\langle S_t\rangle}^q$ of the random variable $\bm X_{t,v,\langle S\rangle}:=f_t(S_t\cup \{v\}, \bm \eta_t)-f_t(S_t,\bm \eta_t)$ conditioned by the observation of $\langle S_t\rangle$. Then, ${\sf Oracle1}$ returns the value ${\sf Oracle1}(q,I,t,\langle S_t\rangle,v):=\tilde{\Delta}_t(v|\langle S_t\rangle):=\sum_{j=1}^q X_{t,v,\langle S_t\rangle}^j/q$. 

Given $\delta,\xi>0$, ${\sf Oracle1}$ satisfies the following property, whose proof is deferred to Appendix \ref{appendix:oracle}.
\begin{proposition}
\label{prop:oracle_property}
 If $q\geq q(\delta,\xi):=2\frac{\Lambda^2}{\delta^2}\ln(2n/\xi)$, then:
 \begin{itemize}
 \item[(i)] any item $\tilde{v}$ maximizing $\tilde{\Delta}_t(\tilde{v}|\langle S_t\rangle)$ verifies $
\mathbb{P}\left[\max_{v\in V}\Delta_t(v|\langle S_t\rangle)-{\Delta_t}(\tilde{v}|\langle S_t\rangle)\leq \delta\right]\geq 1-\xi$;
\item[(ii)] If $q\leq P(q(\delta,\xi))$ for some polynomial $P$, ${\sf Oracle1}$ runs in polynomial time w.r.t. the size of the input instance, $1/\delta$ and $1/\xi$.
\end{itemize}
\end{proposition}

\subsubsection{Oracle2}
\label{def_oracle2}
${\sf Oracle2}$ takes in input an integer $q\geq 1$, an instance $I$ of the SBMSm problem, a polynomial adaptive policy $\pi(t)$ that only works at round $t$ and selects all the items in $n$ steps, an integer $i\in [n]$, and picks $q$ samples $Y_{t,i}^1,\ldots, Y_{t,i}^q$ of the random variable
$\bm Y_{t,i}:=f_t(\{\bm v_1,\ldots, \bm v_{i}\},\bm \eta_t)-f_t(\{\bm v_1,\ldots, \bm v_{i-1}\},\bm \eta_t)$, that denotes the increment of $f_t$ achieved by $\pi(t)$ when selecting the $i$-th item $\bm v_i$ at round $t$, where $\bm v_1,\ldots, \bm v_n$ are the random nodes chosen by $\pi(t)$ in order of selection. Then, ${\sf Oracle2}$ returns the empirical mean ${\sf Oracle2}(q,I,\pi(t),i):=\tilde{\Delta}_{t,i}:=\sum_{j=1}^q Y^j_{t,i}/q$. We have that ${\sf Oracle2}$ satisfies the following property, whose proof is deferred to Appendix \ref{appendix:oracle}.
\begin{proposition}\label{prop:oracle2_property}
 If $q\geq q'(\delta,\xi):=\frac{\Lambda^2}{2\delta^2}\ln(2Tn/\xi)$, then:
 \begin{itemize}
  \item[(i)] $
\mathbb{P}\left[|\Delta_{t,i}-\tilde{\Delta}_{t,i}|\leq \delta,\ \forall t\in [T],i\in [n]\right]\geq 1-\xi,$
where $\Delta_{t,i}$
is
the (exact) expected increment of $f_t$ achieved by $\pi(t)$ when selecting the $i$-th item;
\item[(ii)] if $q\leq P(q'(\delta,\xi))$ for some polynomial $P$, ${\sf Oracle2}$ runs in polynomial time w.r.t. the size of the input instance, $1/\delta$ and $1/\xi$.
 \end{itemize}

\end{proposition}

\subsection{Single-Round Adaptive Policy}
In order to define our approximate policy we first introduce the greedy framework adopted at each round, given the available budget for that round. The
{\em greedy single-round (adaptive) policy} ${\sf \SGR}_{\delta,\xi}(I,t,b)$  (see Algorithm~\ref{unique_round_alg}
for the pseudo-code)
greedily selects $b$ items (in $b$ steps) at round $t$ only, and at each step selects an item maximizing the expected increment of $f$ up to an addend $\delta$ and with prob. at least $1-\xi$,
conditioned by the previous observations;
the item selected at each step can be determined by calling {\sf Oracle1}
for each item $v$ not yet included, and then choosing that maximizing the value ${\sf Oracle1}(q(\delta,\xi),I,t,\langle S_t\rangle,v)$.
\begin{algorithm}[H]
	\caption{${\sf \SGR}_{\delta,\xi}$}
	\label{unique_round_alg}
	\begin{algorithmic}[1]
		\REQUIRE an input instance $I$ of SBMSm, a round $t$, a budget $b\geq 0$, an (unknown) state $\eta_t$ for round $t$;
		\ENSURE a subset $S_t:={\sf \SGR}_{\delta,\xi}(I,t,b,\eta_t)$;
		\STATE let $S_t\leftarrow \emptyset$;
		\WHILE{$|S_t|\leq b$}
		\STATE $v\leftarrow \arg\max_{v'\in V}{\sf Oracle1}(q(\delta,\xi),I,t,\langle S_t\rangle, v')$;\\
		\textcolor{gray}{\% $ \Delta_t(v|\langle S_t,\eta_t\rangle)\geq \max_{v'\in V\setminus S_t} \Delta_t(v'|\langle S_t,\eta_t\rangle)-\delta$ with prob. at least $1-\xi$ by Proposition~\ref{prop:oracle_property}.}
		\STATE $S_t\leftarrow S_t\cup \{v\}$;
		\ENDWHILE
		\RETURN $S_t$.
	\end{algorithmic}
\end{algorithm}

Notice that ${\sf \SGR}_{\delta,\xi}$ is equivalent to an approximate variant of the adaptive greedy policy defined by \cite{golovin2011adaptive}, that
provides a $\left(1-1/e-\epsilon\right)$-approximation to the best
policy when considering
a unique round.


\subsection{Budget Allocation Algorithm}
In order to allocate the budget over rounds, we consider an algorithm called ${\sf BudgetGr}_{\delta,\epsilon}(I)$ (see Algorithm~\ref{BudgetGr_alg} for the pseudo-code), that greedily assigns (in a non-adaptive way) each unit of budget to the round $t$ that maximizes the expected increment of the objective function achieved by the greedy single-round (adaptive) policy; the above expected increment is well-approximated with high probability by algorithm {\sf Oracle2}, in which the single-greedy policy is executed many times to collect a sufficiently high number of samples on the objective function increments.
\begin{algorithm}[H]
	\caption{${\sf BudgetGr}_{\delta,\xi}$}
	\label{BudgetGr_alg}
	\begin{algorithmic}[1]
		\REQUIRE an input instance $I$ of SBMSm;
		\ENSURE a vector $\vec{b}:=(b_1,\ldots, b_T)={\sf BudgetGr}_{\delta,\xi}(I)$;
		\FOR{$t=1,\ldots, T$}
				\STATE 	 $\overline{\Delta}_{t,0}\leftarrow \infty$;
		\FOR{$i=1,\ldots, n$}
\STATE $\tilde{\Delta}_{t,i}\leftarrow {\sf Oracle2}(q'(\delta,\xi),I,{\sf \SGR}_{\delta,\xi}(I,t,n),i)$;
		\STATE $\overline{\Delta}_{t,i}\leftarrow \min\{\overline{\Delta}_{t,i-1},\tilde{\Delta}_{t,i}\}$;
		\ENDFOR
		\STATE 	 $\overline{\Delta}_{t,n+1}\leftarrow -\infty$;\textcolor{gray}{\ \% This is done to impose the unfeasibility of selecting $n+1$ items at each round $t$.}
		\ENDFOR
		\STATE \textcolor{gray}{\% By Proposition~\ref{prop:oracle2_property} applied to ${\sf \SGR}_{\delta,\xi}(I,t,n)$, the event ``$|\Delta_{t,i}-\tilde{\Delta}_{t,i}|\leq \delta,\ \forall t\in [T],i\in [n]$'' holds with prob. at least $1-\xi$, where $\Delta_{t,i}$ denotes the (exact) expected increment of objective function $f_t$, achieved by ${\sf \SGR}_{\delta,\xi}(I,t,n)$ when selecting the $i$-th item, in order of selection. Furthermore, as
		$\Delta_{t,i}$ is
		non-increasing in $i$ up to a small addend (by adaptive-submodularity), we also have that the distance $|\Delta_{t,i}-\overline{\Delta}_{t,i}|$ is small with prob. at least $1-\xi$.}
		\STATE $(b_1,\ldots, b_T)\leftarrow (0,\ldots, 0)$;
		\FOR{$h=1,\ldots, B$}
		\STATE $t^*\leftarrow \arg\max_{t\in [T]}\overline{\Delta}_{t,b_t+1}$;
		\STATE $b_{t^*}\leftarrow b_{t^*}+1$;
		\ENDFOR
		\RETURN $\vec{b}:=(b_1,\ldots, b_T)$;
	\end{algorithmic}
\end{algorithm}

\subsection{Multi-Round Adaptive Policy}\label{subsec:multi}
Let $\vec{b}:=(b_1,\ldots, b_T)$ be the output of ${\sf BudgetGr}_{\delta,\epsilon}(I)$.
The
{\em greedy multi-round (adaptive) policy} ${\sf \MGR}_{\delta,\xi}(I)$ (see Algorithm \ref{multi_round_alg} for the pseudo-code) applies the greedy single-round policy with budget $b_t$ at each round $t\in [T]$ to determine a subset $S_t$, and returns the aggregated subset $\vec{S}=(S_1,\ldots, S_T)$.
 \begin{algorithm}[ht]
	\caption{${\sf \MGR}_{\delta,\xi}$}
	\label{multi_round_alg}
	\begin{algorithmic}[1]
		\REQUIRE an input instance $I$ of SBMSm, an (unknown) aggregated state $\vec{\eta}:=(\eta_1,\ldots, \eta_t)$;
		\ENSURE an aggregated subset $\vec{S}:={\sf \MGR}_{\delta,\xi}(I,\vec{\eta})$;
		\STATE $\vec{b}:=(b_1,\ldots, b_T)\leftarrow {\sf BudgetGr}_{\delta.\xi}(I)$;
		\FOR{$t=1,\ldots, T$}
		\STATE $S_t\leftarrow {\sf \SGR}_{\delta,\xi}(I,t,b_t,\eta_t)$;
		\STATE wait the end of round $t$;
		\ENDFOR
		\RETURN $\vec{S}:=(S_1,\ldots, S_T)$.
	\end{algorithmic}
\end{algorithm}

\subsection{Approximation Factor and Time-Complexity of the Policy}\label{subsec:approx}
Theorem \ref{thm1} shows that, for any fixed $\epsilon>0$ and opportune choices of $\delta$ and $\xi$, ${\sf \MGR}_{\delta,\xi}(I)$ achieves a $1/2(1-1/e-\epsilon)$ approximation to the optimal adaptive policy and runs in polynomial-time. In the following remark, we first explain why the approaches previously adopted to deal with the single-round case \cite{golovin2011adaptive} could not be directly applied to the general multi-round framework to obtain a constant approximation, thus showing that a new approach like ours was somehow necessary.
\begin{remark}\label{remaadaptive}
As the adaptive greedy algorithm of \cite{golovin2011adaptive} returns a $(1-1/e-\epsilon)$-approximation for (adaptive) single-round submodular maximization problems, one could reasonably believe that their approach guarantees the same approximation for the multi-round generalization considered in this work; if this was the case, our partially adaptive algorithm would not guarantee a better approximation than their fully adaptive algorithm, and it would be of interest only because partial adaptivity is easier to implement in various realistic contexts. However, a direct application of the adaptive greedy algorithm of \cite{golovin2011adaptive} would require to represent the input multi-round instance as a union of $T$ independent single-round sub-instances (i.e., one for each round) and, similarly as in the greedy-single round policy (see Algorithm \ref{unique_round_alg}), would greedily and adaptively select $B$ items (in $B$ steps), possibly moving back and forth between rounds. Despite our model allows to move from one round to the subsequent ones, it forbids to return to selecting items in previously abandoned rounds (as this would mean go back in time and cancel the choices previously made). Thus, due to this constraint, the greedy choices suggested by the algorithm of \cite{golovin2011adaptive} may not be valid when applied to our multi-round framework. 

Furthermore, if we apply the algorithm of \cite{golovin2011adaptive} by restricting the greedy choice, at each step, to items of the current or the subsequent rounds, the resulting approximation could be arbitrarily low (i.e., inefficient). To show this, we consider a simple instance $I$ defined as follows: we have two rounds, a set of $n$ items and a budget $B:=n+1$; a special item $v^*$ has weight weight 1 for each round, and the remaining items have weight $1/2$ at round $1$ and weight $0$ at round $2$; the objective function, at each round, simply counts the total weight of selected items. We observe that the considered instance can be easily modelled as a (fully deterministic) binary submodular stochastic probing problem (see Section~\ref{sec:appl}). An optimal policy for $I$ requires to select all items at round $1$ and the special item at round $2$, and then guarantees a value of $(n-1)/2+2=n/2+3/2$. However, the solution returned by the considered adaptation of the  greedy algorithm of \cite{golovin2011adaptive} selects the special item at round 1, moves to the subsequent round and then selects all the available $n$ items (as it cannot go back to round 1); in such a case the obtained value is $2$, and then only guarantees a $2/(n/2+3/2)=4/(n+3)=O(1/n)$-approximation of the optimal value, that becomes arbitrarily low as $n$ increases.
\end{remark}
\begin{theorem}\label{thm1}
Let $I=(T,B,V,H,(\mathcal{P}_t)_{t\in [T]},R,(f_t)_{t\in [T]})$ be an input instance of SBMSm.
The greedy multi-round policy ${\sf \MGR}_{\delta,\xi}(I)$ is a $1/2(1-1/e-\epsilon)$-approximate policy for $\delta:=\xi:=\frac{\lambda c\epsilon}{B(4+3\Lambda)}$ and any constant $c$ such that $0<c\leq 1$, and it is polynomial in the size of $I$ and $1/\epsilon$.
\end{theorem}
\begin{proof}[Proof Sketch]
We first show the polynomial-time complexity of the algorithm. First of all, by Propositions~ \ref{prop:oracle_property} and \ref{prop:oracle2_property} we have that both {\sf Oracle1} and {\sf Oracle2} run in polynomial time w.r.t. the size of $I$, $1/\delta$ and $1/\xi$. Then, as $1/\epsilon=O\left(\frac{nT\Lambda}{\lambda\delta}\right)=O\left(\frac{nT\Lambda}{\lambda\xi}\right)$, by the polynomial boundedness property (that polynomially relates quantities $\Lambda$ and $1/\lambda$ to the size of $I$) we conclude that {\sf Oracle1} and {\sf Oracle2} are polynomial in $1/\epsilon$ and the size of $I$. Finally, the polynomial-time complexity of the whole greedy multi-round policy follows from the fact that such policy applies $n\cdot T$ times algorithm {\sf Oracle2} to compute the budget vector (via algorithm ${\sf BudgetGr}_{\delta,\xi}(I)$) and $\sum_{t\in [T]}b_t=B < n\cdot T$ times the algorithm {\sf Oracle1} to select $b_t$ items at each round $t$ (via algorithm ${\sf \SGR}_{\delta,\xi}(I,t,b_t)$); the other steps of the greedy multi-round policy are less significant in terms of time-complexity.

In the remainder of the proof, we focus on the analysis of the approximation factor. For simplicity, we first  show that the greedy multi-round policy achieves a $1/2(1-1/e)$-approximation (that is, with $\epsilon=0$), assuming that the probabilistic oracles are exact, i.e., if the algorithms {\sf Oracle1} and {\sf Oracle2} provided in Section \ref{sec:sampling} can be run with $\delta:=\xi:=0$. The proof of the approximation factor for arbitrary $\delta=\xi>0$ extends the proof arguments presented below, and is given in Appendix~\ref{appendix:thm1}.

Let $\pi^*$ be an optimal adaptive policy for $I$, and let $OPT$ denote its expected value. Given a round $t\in [T]$ and a real number $d\geq 0$, let $\overline{OPT}_t(d)$ denote the optimal expected value achievable by a single-round policy at round $t$ assuming that the expected number of selected items is $d$. Instead, for an integer $b\geq 0$, let $OPT_t(b)$ denote the optimal expected value at round $t$ assuming that the number of selected items is equal to $b$, independently on the particular realization; we extend the above definition to any real number $d\geq 0$ by setting $OPT_t(d):=(d-\lfloor d\rfloor)OPT_t(\lceil d\rceil )+(1-(d-\lfloor d\rfloor))OPT_t(\lfloor d\rfloor )$, i.e., $OPT_t(d)$ is the expected value of an adaptive policy that optimally selects $\lceil d\rceil$ with probability $p:=d-\lfloor d\rfloor$, and optimally selects $\lfloor d\rfloor$ items with probability $1-p$. We observe that the above policy selects $d$ items in expectation, as $p\lceil d\rceil+(1-p)\lfloor d\rfloor=(d-\lfloor d\rfloor)\lceil d\rceil+(\lceil d\rceil-d)\lfloor d\rfloor=d(\lceil d\rceil-\lfloor d\rfloor)=d$ if $d\notin \mathbb{Z}_{\geq 0}$, and $p\lceil d\rceil+(1-p)\lfloor d\rfloor=pd+(1-p)d=d$ if $d\in \mathbb{Z}_{\geq 0}$; furthermore, we observe that $OPT_t(d)\leq \overline{OPT}_t(d)$ holds for any $d\geq 0$.

For any $t\in [T]$, let $d^*_t$ denote the expected number of items selected by $\pi^*$ at round $t$, and let $OPT_t$ denote the expected value of objective function $f_t$ under policy $\pi^*$. By definition of $\overline{OPT}_t(d^*_t)$, and since states and objective functions of distinct rounds are independent, we have that
\begin{equation}\label{thm1:eq1}
OPT\leq \sum_{t\in [T]}\overline{OPT}_t(d^*_t).
\end{equation}
Given an integer $b\geq 0$, let $GR_t(b)$ denote the expected value of the greedy single-round policy ${\sf \SGR}_{\delta,\xi}$ (see Algorithm \ref{unique_round_alg}) applied to round $t$, budget $b$ and $\delta=\xi=0$ (i.e., by resorting to exact oracles). Let $\vec{b}=(b_1,\ldots, b_T)$ denote the budget vector computed by algorithm ${\sf BudgetGr}_{\delta,\xi}$ (see Algorithm \ref{BudgetGr_alg}) with $\delta=\xi=0$, and let $GR(\vec{b})$ denote the expected value of the greedy multi-round greedy policy ${\sf \MGR}_{\delta,\xi}$ applied to budget vector $\vec{b}$, again with $\delta=\xi=0$. We observe that
\begin{equation}\label{thm1:eq2}
GR(\vec{b})=\sum_{t\in [T]}GR_t(b_t).
\end{equation}
To show the desired approximation factor, we break down the proof into four main parts:
\begin{enumerate}
\item By resorting to \eqref{thm1:eq1} and Lemma \ref{lem1} below, we have
\begin{equation}\label{thm1:eq3}
\frac{1}{2}\cdot OPT\leq \sum_{t\in [T]}\frac{1}{2}\cdot \overline{OPT}_t(d^*_t)\leq \sum_{t\in [T]}OPT_t(d^*_t).
\end{equation}
We stress that the proof of this step uses the strong submodularity property (in particular, it is used in the proof of Lemma \ref{lem1} below). 
\item By derandomizing vector $\vec{d}^*:=(d^*_1,\ldots, d^*_T)$, we have that there exists a budget vector $\vec{b^*}=(b_1^*,\ldots, b_T^*)$ of non-negative integers with $\sum_{t\in [T]}b^*_t=B$ such that 
\begin{equation}\label{thm1:eq4}
\sum_{t\in [T]}OPT_t(d^*_t)\leq \sum_{t\in [T]}OPT_t(b^*_t)
\end{equation}
(see Lemma \ref{lem2} below). 
\item Then, by exploiting the $(1-1/e)$-approximation of the standard adaptive greedy algorithm \cite{golovin2011adaptive} (that resorts to the adaptive submodularity property), we have 
\begin{equation}\label{thm1:eq5}
\sum_{t\in [T]}\left(1-\frac{1}{e}\right)OPT_t(b_t^*)\leq \sum_{t\in [T]}GR_t(b^*_t)=GR(\vec{b}^*)
\end{equation}
(see Lemma \ref{lem3} below). 
\item Finally, by exploiting the greedy assignment of ${\sf BudgetGr}_{0,0}$ and the adaptive submodularity, we have 
\begin{equation}\label{thm1:eq6}
GR(\vec{b}^*)\leq GR(\vec{b}).
\end{equation}
(see Lemma \ref{lem4} below).
\end{enumerate}
By combining the above inequalities we obtain 
\begin{align}
\frac{1}{2}\left(1-\frac{1}{e}\right)OPT&\leq \left(1-\frac{1}{e}\right)\sum_{t\in [T]}OPT_t(d^*_t)\leq GR(\vec{b}^*)\leq GR(\vec{b}),
\end{align}
and this shows the claim.
\end{proof}

\subsection{Technical Lemmas} Below are the lemmas cited in the proof sketch of Theorem \ref{thm1}.
\begin{lemma}\label{lem1}
For any round $t\in [T]$ and real number $d\geq 0$, we have $\overline{OPT}_t(d)\leq 2\cdot OPT_t(d)$. 
\end{lemma}
\begin{proof}[Proof Sketch]
We first assume that $d=b$ for some integer $b\geq 0$. The proof for a generic real number $d\geq 0$ is given in Appendix \ref{appendix:thm1}. 
If $b=0$ the claim trivially follows, as  $\overline{OPT}_t(0)=OPT_t(0)=0$. Thus, we assume that $b\geq 1$. Let $\pi^*_t(b)$ be an arbitrary optimal adaptive policy at round $t$ that selects $b$ items in expectation. 
 For any $i\in [n]$, let ${\bm S_i'}$ denote the (possibly empty) random set of at most $b$ items selected by $\pi^*_t(b)$ after having already selected $b\cdot (i-1)$ items (that is, $\sqcup_{i=1}^n\bm S_i'$ is the set of items selected by the policy $\pi^*_t(b)$), $A_i$ denote the event ``$|{\bm S_i'}|> 0$'' and $\Delta(\pi^*_t(b)|A_i)$ denote the expected increment of $f_t$ gained from the selection of ${\bm S_i'}$ (after the set $\sqcup_{j=1}^{i-1}\bm S_j'$ has already been selected) and conditioned by event $A_i$; furthermore, for $i\geq 2$, let $C_{i}$ denote the event ``$|{\bm S_{i-1}'}|= b$''. We have
 \begin{equation}\label{lem1:ineq0}
 \overline{OPT}_t(b)=\sigma(\pi^*_t(b))=\sum_{i=1}^{n} \mathbb{P}[A_i]\cdot\Delta(\pi^*_t(b)|A_i)\leq \sum_{i=1}^{n} \mathbb{P}[A_i]\cdot OPT_t(b),
 \end{equation}
 where: the second equality is obtained by decomposing the value of policy $\sigma(\pi^*_t(b))$ into a sum of expected increments and by observing that $\bm S_i'$ is non-empty only if event $A_i$ is true; the last inequality holds by the strong adaptive submodularity. Furthermore, we have that\footnote{We use $\#X$ in place of $|X|$ to denote the cardinality of a set $X$. This is done to avoid confusion when  symbol $|$ is used for conditional expectation.}
 \begin{align}
 b&\geq\sum_{i=2}^{n}\mathbb{E}_{\bm \eta_t}[\#{\bm S_{i-1}'}]\nonumber\\
 &\geq \sum_{i=2}^{n} \mathbb{P}[C_i]\cdot \mathbb{E}_{\bm \eta_t}[\#{\bm S_{i-1}'} \ |\ C_i]\nonumber\\
 &=\sum_{i=2}^{n} \mathbb{P}[C_i]\cdot b\label{lem1:ineq1}\\
  &\geq \sum_{i=2}^{n} \mathbb{P}[A_{i}]\cdot b\label{lem1:ineq2}\\
 &\geq -b+\sum_{i=1}^n \mathbb{P}[A_{i}]\cdot b\label{lem1:ineq3},
 \end{align}
 where: \eqref{lem1:ineq1} holds since, by definition of event $C_i$, the cardinality of ${\bm S_{i-1}'}$ is constantly equal to $b$ if event $C_i$ is true; \eqref{lem1:ineq2} holds since $\mathbb{P}[C_i]\geq  \mathbb{P}[A_{i}]$ (as event $A_{i}$ implies event $C_i$). By rearranging \eqref{lem1:ineq3}, we obtain 
 \begin{equation}\label{lem1:ineq3.1}
\sum_{i=1}^n \mathbb{P}[A_{i}]\leq 2,
\end{equation}  
and by applying this inequality to \eqref{lem1:ineq0} we obtain 
\begin{equation}\label{lem1:ineq4}
\overline{OPT}_t(b)\leq \sum_{i=1}^{n} \mathbb{P}[A_i]\cdot OPT_t(b)\leq 2\cdot OPT_t(b),
\end{equation}
that shows the claim.
\end{proof}
\begin{lemma}\label{lem2}
For any vector $\vec{d}=(d_1,\ldots, d_T)$ of non-negative real numbers such that $\sum_{t\in [T]}d_t=B$, there exists a budget vector $\vec{b^*}=(b_1^*,\ldots, b_T^*)$ of non-negative integers with $\sum_{t\in [T]}b^*_t=B$ such that 
\begin{equation*}
\sum_{t\in [T]}OPT_t(d_t)\leq \sum_{t\in [T]}OPT_t(b^*_t).
\end{equation*}
\end{lemma}
\begin{proof}
Let $\vec{d}=(d_1,\ldots, d_T)$ be a vector of non-negative real numbers such that $\sum_{t\in [T]}d_t=B$. For any $t\in [T]$, we have that $OPT_t(d_t)=p_t\cdot OPT_t(\lceil d_t\rceil )+(1-p_t)\cdot OPT_t(\lfloor d_t\rfloor )$, for $p_t=(d_t-\lfloor d_t\rfloor)$. We say that a round $t$ is integral if $d_t$ is integer, while it is fractional if $d_t$ is not. We observe that $p_t=0$ if $t$ is integral and $p_t\in (0,1)$ otherwise. 

If all rounds are integral, we directly can set $\vec{b^*}=\vec{d^*}$ to obtain the claim. Henceforth, we assume that there exists at least a fractional round $t_1\in [T]$. We claim that there must necessarily be another fractional round $t_2\neq t_1$. Indeed, if all rounds $t\neq t_1$ were integral we would have that each $d_t$, with $t\neq t_1$, would be integer. As $\sum_{t\in [T]}d_t=B$, we have that $d_{t_1}=B-\sum_{t\neq t_1}d_t$. Thus, $d_{t_1}$ must be an integer, as it is defined by sums and differences of integer numbers. But this contradicts the assumption that $t_1$ is fractional. We conclude that that there exist at least two fractional rounds $t_1,t_2$. 

Assume w.l.o.g. that $0<d_{t_1}-\lfloor d_{t_1}\rfloor\leq \lceil d_{t_2}\rceil-d_{t_2}<1$ (otherwise, it is sufficient to permute indexes $t_1$ and $t_2$). For any $\beta\in[-(\lceil d_{t_2}\rceil-d_{t_2}), d_{t_1}-\lfloor d_{t_1}\rfloor]$, let $\vec{d}(\beta)$ be the vector obtained by moving an amount $\beta$ (resp. $-\beta$) from $d_{t_1}$ to $d_{t_2}$ if $\beta\geq 0$ (resp. from $d_{t_2}$ to $d_{t_1}$ if $\beta< 0$), i.e., $d_{t_1}(\beta)=d_{t_1}-\beta$, $d_{t_2}(\beta)=d_{t_2}+\beta$ and $d_{t}(\beta)=d_t$ for any $t\notin \{t_1,t_2\}$. Let $\Delta_{\beta}:=\beta (OPT_{t_2}(\lceil d_{t_2}\rceil)-OPT_{t_2}(\lfloor d_{t_2}\rfloor)-OPT_{t_1}(\lceil d_{t_1}\rceil)+OPT_{t_1}(\lfloor d_{t_1}\rfloor))$.
For any $\beta\in[-(\lceil d_{t_2}\rceil-d_{t_2}), d_{t_1}-\lfloor d_{t_1}\rfloor]$, we have 
\begin{align}
&OPT(\vec{d}(\beta))\nonumber\\
&=(d_{t_2}+\beta-\lfloor d_{t_2}+\beta\rfloor)OPT_{t_2}(\lceil d_{t_2}+\beta\rceil)+(\lceil d_{t_2}+\beta\rceil-(d_{t_2}+\beta))OPT_{t_2}(\lfloor d_{t_2}+\beta\rfloor)\nonumber\\
&\quad  + (d_{t_1}-\beta-\lfloor d_{t_1}-\beta\rfloor)OPT_{t_1}(\lceil d_{t_1}-\beta\rceil)+(\lceil d_{t_1}-\beta\rceil-(d_{t_1}-\beta))OPT_{t_1}(\lfloor d_{t_1}-\beta\rfloor)\nonumber\\
&\quad  +\sum_{t\notin\{t_1,t_2\}}OPT_t(d_t)\nonumber\\
&=(d_{t_2}+\beta-\lfloor d_{t_2}\rfloor)OPT_{t_2}(\lceil d_{t_2}\rceil)+(\lceil d_{t_2}\rceil-(d_{t_2}+\beta))OPT_{t_2}(\lfloor d_{t_2}\rfloor)\nonumber\\
&\quad  + (d_{t_1}-\beta-\lfloor d_{t_1}\rfloor)OPT_{t_1}(\lceil d_{t_1}\rceil)+(\lceil d_{t_1}\rceil-(d_{t_1}-\beta))OPT_{t_1}(\lfloor d_{t_1}\rfloor)\nonumber\\
&\quad  +\sum_{t\notin\{t_1,t_2\}}OPT_t(d_t)\label{lem2:ineq10}\\
&=\beta (OPT_{t_2}(\lceil d_{t_2}\rceil)-OPT_{t_2}(\lfloor d_{t_2}\rfloor)-OPT_{t_1}(\lceil d_{t_1}\rceil)+OPT_{t_1}(\lfloor d_{t_1}\rfloor))\nonumber\\
&\quad +(d_{t_2}-\lfloor d_{t_2}\rfloor)OPT_{t_2}(\lceil d_{t_2}\rceil )+(\lceil d_{t_2}\rceil-d_{t_2})OPT_{t_2}(\lfloor d_{t_2}\rfloor )\nonumber\\
&\quad +(d_{t_1}-\lfloor d_{t_1}\rfloor)OPT_{t_1}(\lceil d_{t_1}\rceil )+(\lceil d_{t_1}\rceil-d_{t_1})OPT_{t_1}(\lfloor d_{t_1}\rfloor )\nonumber\\
&\quad +\sum_{t\notin\{t_1,t_2\}}OPT_t(d_t)\nonumber\\
&=\Delta_{\beta}+OPT(\vec{d}),\label{lem2:ineq1}
\end{align}
where \eqref{lem2:ineq10} holds since, by rearranging the assumptions $\beta\in[-(\lceil d_{t_2}\rceil-d_{t_2}), d_{t_1}-\lfloor d_{t_1}\rfloor]$ and $0<d_{t_1}-\lfloor d_{t_1}\rfloor\leq \lceil d_{t_2}\rceil-d_{t_2}<1$, we have $\lfloor d_{t_2}\rfloor\leq d_{t_2}+\beta\leq \lceil d_{t_2}\rceil$ and $\lfloor d_{t_1}\rfloor\leq d_{t_1}-\beta\leq \lceil d_{t_1}\rceil$; furthermore, we trivially have $d_t(\beta)\geq 0$ for any $t\in [T]$ and $\sum_{t\in [T]}d_t(\beta)=B$.
Finally, since $\Delta_{\beta}=0$ for $\beta=0$, we have that the maximum value of $\Delta_{\beta}$ over $\beta\in [-(\lceil d_{t_2}\rceil-d_{t_2}), d_{t_1}-\lfloor d_{t_1}\rfloor]$ is non negative; furthermore, since $\Delta_{\beta}$ is linear in $\beta \in [-(\lceil d_{t_2}\rceil-d_{t_2}), d_{t_1}-\lfloor d_{t_1}\rfloor]$, its maximum value is achieved by some $\beta\in \{-(\lceil d_{t_2}\rceil-d_{t_2}), d_{t_1}-\lfloor d_{t_1}\rfloor\}$. We conclude that there exists $\beta\in \{-(\lceil d_{t_2}\rceil-d_{t_2}), d_{t_1}-\lfloor d_{t_1}\rfloor\}$ such that $\Delta_{\beta}\geq 0$, and by \eqref{lem2:ineq1} we get
\begin{equation}\label{lem2:ineq2}
OPT(\vec{d}(\beta))=\Delta_{\beta}+OPT(\vec{d})\geq OPT(\vec{d});
\end{equation}
furthermore, as $\beta\in \{-(\lceil d_{t_2}\rceil-d_{t_2}), d_{t_1}-\lfloor d_{t_1}\rfloor\}$, we have that at least one round $t\in \{t_1,t_2\}$ becomes integral in $\vec{d}(\beta)$, and the rounds that were previously integral in $\vec{d}$ continue to be integral in $\vec{d}(\beta)$. Then, the number of integral rounds of $\vec{d}(\beta)$ is higher than that related to $\vec{d}$. 

We conclude that, by iteratively applying the above process to $\vec{d}':=\vec{d}(\beta)$, after at most $T$ iterations we obtain a vector $\vec{b}^*$ satisfying the claim. 
\end{proof}
\begin{lemma}\label{lem3}
For any $t\in [T]$ and integer $b\geq 0$, we have $\left(1-\frac{1}{e}\right)OPT_t(b)\leq GR_t(b)$.
\end{lemma}
\begin{proof}[Proof Sketch]
Fix $t\in [T]$ and $b\geq 0$. The greedy single-round policy, applied to budget $b$ and with exact oracles, is equivalent to the adaptive greedy algorithm of \cite{golovin2011adaptive} that greedily selects $b$ items in the single-round instance associated with round $t$. They showed that, under adaptive submodularity, their algorithm achieves a $(1-1/e)$-approximation to the optimum. Thus, by turning their result in our setting, we obtain the claim.\footnote{Despite the analysis of the adaptive greedy algorithm (that is equivalent to the greedy single-round policy) has been already provided by \cite{golovin2011adaptive}, it can be easily derived from the proof of Lemma \ref{lem5} in Appendix \ref{appendix:thm1}. Indeed, in the proof of Lemma \ref{lem5} we provide the analysis of the greedy single-round policy for the case of approximated oracles (i.e., $\delta,\xi>0$), and the simpler analysis for the case of approximated oracles can be obtained by setting $\delta=\xi=0$ in its proof.}
\end{proof}
\begin{lemma}\label{lem4}
For any vector $\vec{b}'$ of non-negative integers such that $\sum_{t\in [t]}b_t'=B$, we have that 
$GR(\vec{b}')\leq GR(\vec{b})$, where $\vec{b}$ is the vector returned by ${\sf BudgetGr}_{0,0}$.
\end{lemma}
\begin{proof}
for any $t\in [T]$ and $i\in [n]$, let $\Delta_{t,i}$ denote the expected increment gained from the $i$-th item (in order of selection), when applying the greedy single-round policy at round $t$ that selects $n$ items under exact oracles (i.e., with $\delta=\xi=0$). Let $\bm S_{t,i}$ denote the random set of the first $i$ items selected at round $t$ by the greedy single-round policy, and let $\bm v_{t,i}$ the $i$-th random item selected at round $t$ (that is, the item $v$ maximizing $\Delta_t(v|\langle \bm S_{t,{i-1}}\rangle)$). For any integers $i,j\in [n]$ with $i\geq j$, we have
\begin{align}
\Delta_{t,j}&=\mathbb{E}_{\langle \bm S_{t,{j-1}}\rangle}\left[\max_{v\in V}\Delta_t(v|\langle \bm S_{t,{j-1}}\rangle)\right]\nonumber\\
&=\mathbb{E}_{\langle \bm S_{t,{i-1}}\rangle}\left[\max_{v\in V}\Delta_t(v|\langle \bm S_{t,{j-1}}\rangle)\right]\nonumber\\
&\geq \mathbb{E}_{\langle \bm S_{t,{i-1}}\rangle}\left[\mathbb{E}_{\bm v_{t,i}}\left[\Delta_t(\bm v_{t,i}|\langle \bm S_{t,{j-1}}\rangle)\right]\right]\label{lem4:ineq1}\\
&\geq \mathbb{E}_{\langle \bm S_{t,{i-1}}\rangle}\left[\mathbb{E}_{\bm v_{t,i}}\left[\Delta_t(\bm v_{t,i}|\langle \bm S_{t,{i-1}}\rangle)\right]\right]\label{lem4:ineq2}\\
&=\mathbb{E}_{\langle \bm S_{t,{i-1}}\rangle}\left[\max_{v\in V}\Delta_t(v|\langle \bm S_{t,{i-1}}\rangle)\right]\nonumber\\
&=\Delta_{t,i},\label{lem4:ineq3}
\end{align}
where \eqref{lem4:ineq1} holds by the greedy choice of the greedy single-round policy and \eqref{lem4:ineq2} holds by the adaptive submodularity (as $\langle \bm S_{t,j-1}\rangle\prec \langle \bm S_{t,i-1}\rangle$). Thus, 
we have that, for any $t\in [T]$, $\Delta_{t,i}$ is non-increasing in $i\in [n]$. 

Now, fix a vector $\vec{b}'$ of non-negative integers such that $\sum_{t\in [t]}b_t'=B$. As $GR_t(b)$ can be written as $\sum_{i\in [b]}\Delta_{t,i}$ for any integer $b\geq 0$, we have that
\begin{equation}\label{lem4:ineq4}
GR(\vec{b}')=\sum_{t\in [t]}GR_t(b_t')=\sum_{t\in [T]}\sum_{i\in [b'_t]}\Delta_{t,i}
\end{equation}
and 
\begin{equation}\label{lem4:ineq5}
GR(\vec{b})=\sum_{t\in [t]}GR_t(b_t)=\sum_{t\in [T]}\sum_{i\in [b_t]}\Delta_{t,i}.
\end{equation}
By exploiting the greedy construction of $\vec{b}$, and the fact that each $\Delta_{t,i}$ is non-increasing in $i$, we can show that
\begin{equation}\label{lem4:ineq6}
\sum_{t\in [T]}\sum_{i\in [b'_t]}\Delta_{t,i}\leq \sum_{t\in [T]}\sum_{i\in [b_t]}\Delta_{t,i}.
\end{equation}
Indeed, assume that $\vec{b'}\neq \vec{b}$ (otherwise, inequality \eqref{lem4:ineq6} trivially holds as equality), and let $t_1,t_2\in T$ be two distinct rounds such that $b_{t_1}>b'_{t_1}$ and $b_{t_2}<b'_{t_2}$. Consider the iteration of ${\sf BudgetGr}_{0,0}$ in which $b_{t_1}'$ units of budget have been already assigned to round $t_1$ and such round receives a further unit of budget, and let $i$ be the budget already assigned to round $t_2$ before that iteration. We have $\Delta_{t_1,b_{t_1}'+1}\geq \Delta_{t_2,i+1}\geq \Delta_{t_2,b_{t_2}+1}\geq \Delta_{t_2,b_{t_2}'}$, where the first inequality follows by the greedy choice of ${\sf BudgetGr}_{0,0}$ and the remaining ones from the non-increasing monotonicity of each $\Delta_{t_2,j}$ w.r.t. to $j$. Thus, if move one unit of budget from round $t_2$ to round $t_1$ in vector $\vec{b'}$ we obtain a vector $\vec{b''}$ such that $\sum_{t\in [T]}\sum_{i\in [b'_t]}\Delta_{t,i}\leq \sum_{t\in [T]}\sum_{i\in [b''_t]}\Delta_{t,i}$. By iterating the above process at most $B$ times, we reach a vector equal to $\vec{b}$, thus obtaining a chain of inequalities leading to \eqref{lem4:ineq6}.

To conclude the proof, we combine inequalities \eqref{lem4:ineq4}, \eqref{lem4:ineq5} and \eqref{lem4:ineq6} and we obtain the claim. 
\end{proof}

\section{Budget-adaptivity Gap}

In this section, we introduce and quantify the budget-adaptivity gap, to measure how much a fully adaptive optimal policy is better than an optimal partially adaptive policy that, as in the greedy multi-round policy considered above (see Section \ref{subsec:multi}), allocates the budget in advance.
\begin{definition}\label{def:bud_ad_gp}
Given an instance $I$ of SBMSm, the {\em budget-adaptivity gap} of $I$ is defined as $BGap(I):=\sigma(\pi^*)/\sigma(\pi)$, where $\pi^*$ is an optimal adaptive policy for $I$,  and $\pi$ is a policy that is optimal among all partially adaptive policies in which we first choose the budget $b_t$ to be spent for each round $t$ (i.e., non-adaptively) and then we adaptively select $b_t$ items at each round $t$. Furthermore, the budget-adaptivity gap of  SBMSm is defined as $BGap:=\sup_{I\in \text{SBMSm}}BGap(I)$. 
\end{definition}
In the following theorem, we provide a constant upper bound on the adpativity gap, thus showing that the fully adaptive optimum is, in general, strictly higher than the partially adaptive one, but it can still be approximated well by resorting to partial adaptivity.
\begin{theorem}\label{thmgapupp}
The budget-adaptivity gap of SBMSm is at most $2$.
\end{theorem}
\begin{proof}
Let $I$ be an arbitrary input instance of SBMSm and ${\sf \MGR}_{0,0}(I)$ denote the multi-round greedy policy defined in Section \ref{subsec:multi}, assuming that the probabilistic oracles ({\sf Oracle1} and {\sf Oracle2} of Section \ref{sec:sampling}) work with $\delta=\xi=0$, that is, they are able to return with probability $1$ the exact expected increment of the objective functions. Since the proof sketch of Theorem \ref{thm1} given above holds for $\delta=\xi=0$, we can reuse some of its parts to show this theorem. Let $OPT$ denote the optimal fully adaptive value, $OPT_t(b)$ denote the optimal single-round policy at round $t$ with budget $b\geq 0$, $OPT(\vec{b})$ denote the optimal partially adaptive value if the budget is assigned according to a vector $\vec{b}$, and $OPT_P$ denote the optimal value achievable under partially adaptive policies. 
By combining Parts 1 and 2 of the proof sketch of Theorem \ref{thm1}, we have that there exists a vector $\vec{b}^*$ of non-negative integers with $\sum_{t\in [T]}b_t^*=B$ such that 
\begin{equation}\label{thmgapupp:ineq1}
\frac{1}{2}\cdot OPT\leq \sum_{t\in [T]}OPT_t(b_t^*)=OPT(\vec{b}^*)\leq OPT_P
\end{equation}
where the first inequality is obtained by combining \eqref{thm1:eq3} and \eqref{thm1:eq4}, the first equality holds since the optimal policy partially adaptive policy constrained by the budget assignment $\vec{b}^*$ is achieved by optimizing the adaptive selection at each round with the available budget, and the last inequality holds by definition of optimal partially adaptive policies. By \eqref{thmgapupp:ineq1}, the claim follows. 
\end{proof}
In the following theorem, we provide a lower bound that almost matches the above result. 
\begin{theorem}\label{thmgaplow}
The budget-adaptivity gap of SBMSm is at least $e/(e-1)\approx 1.582$.
\end{theorem}
\begin{proof}

We equivalently show that, for any $\epsilon>0$, there exists an instance $I(\epsilon)$ of SBMSm parametrized by $\epsilon$ such that $BGap(I(\epsilon))> e/(e-1)-\epsilon$; then, by the arbitrariness of $\epsilon$ we obtain $BGap\geq e/(e-1)$.

For a fixed $\epsilon>0$, consider an instance $I(\epsilon)$ of the multi-round binary stochastic probing problem (see Section \ref{sec:appl}) defined as follows: we have $T$ rounds and a set $V$ of $n:=T^{3/2}$ items, where $T:=T(\epsilon)$ is a sufficiently large perfect square that will be fixed later; the budget is $B:=n=T^{3/2}$; each item $v\in V$ is active with probability $p:=1/\sqrt{T}$, for each round $t\in [T]$; the objective function $g_t$ of each round $t$ is equal to $1$  if the number of active items is at least $1$, and $0$ otherwise.

In the following, we will provide a lower bound on $BGap(I(\epsilon))$, by comparing the value of an optimal partially adaptive policy $\pi$ and a specific fully adaptive policy $\pi'$ defined below. If we show that $\sigma(\pi')/\sigma(\pi)> \frac{e}{e-1}-\epsilon$, this implies that $BGap(I(\epsilon))\geq \sigma(\pi')/\sigma(\pi)> \frac{e}{e-1}-\epsilon$, that is, the desired lower bound.

Let $\pi$ be the partially adaptive policy that non-adaptively selects the first $B/T=\sqrt{T}$ items at each round. Let $\pi'$ be the fully adaptive policy that, at each round, repedeatly select all items until finding an active one, and then move to the subsequent round; the execution ends if all the budget has been spent or if there are no more rounds available. We have that
\begin{align}
\sigma(\pi)&=\sum_{t\in [T]}\mathbb{P}[\text{$\pi$ has selected at least an active item at round $t$}]\nonumber\\
&=\sum_{t\in [T]}\left(1-\mathbb{P}[\text{$\pi$ has selected no active item at round $t$}]\right)\nonumber\\
&=\sum_{t\in [T]}\left(1-\prod_{k=1}^{B/T}\mathbb{P}[\text{the $k$-th item selected at round $t$ by $\pi$ is not active}]\right)\label{thmgaplow_eq1}\\
&=\sum_{t\in [T]}\left(1-\prod_{k=1}^{\sqrt{T}}(1-p)\right)= T\cdot \left(1-\left(1-\frac{1}{\sqrt{T}}\right)^{\sqrt{T}}\right),\label{thmgaplow_eq2}
\end{align}
where \eqref{thmgaplow_eq1} holds by the independence of the items activation process.

Now, we determine a lower bound on $\sigma(\pi')$. For any $k\in [B]$, let $\bm X_k$ be an independent Bernoulli variable with parameter $p$. We observe that $\min\{T,\sum_{k\in [B]}\bm X_k\}$ is distributed as the value of policy $\pi'$. Indeed, $\pi'$ is equivalent, in terms of total revenue, to a single-round policy that repeatedly selects $B$ items, each one independently active with probability $p=1/\sqrt{T}$, counts the number of selected active items and returns the maximum between such number and $T$. 

For $t:=\sqrt{T}$, we have
\begin{align}
\sigma(\pi')&\geq \mathbb{P}[\text{the value of $\pi'$ is at least $T-t$}]\cdot (T-t)=\mathbb{P}[\text{$\pi'$ selects at least $T-t$ active items}]\cdot (T-t)\nonumber\\
&=\mathbb{P}\left[\min\left\{T,\sum_{k\in [B]}\bm X_k\right\}\geq T-t\right]\cdot (T-t)\label{thmgaplow_eq3}\\
&=\mathbb{P}\left[\sum_{k\in [B]}\bm X_k\geq T-t\right]\cdot (T-t)=\mathbb{P}\left[\sum_{k\in [B]}\bm X_k-\mathbb{E}\left[\sum_{k\in [B]}\bm X_k\right]\geq -t\right]\cdot (T-t)\label{thmgaplow_eq4}\\
&\geq \left(1-\frac{1}{e^{2t^2/B}}\right)\cdot (T-t)\label{thmgaplow_eq5}\\
&= \left(1-\frac{1}{e^{2/\sqrt{T}}}\right)\cdot (T-\sqrt{T}),\label{thmgaplow_eq6}
\end{align}
where \eqref{thmgaplow_eq3} holds by the above observations, \eqref{thmgaplow_eq4} follows from $\mathbb{E}\left[\sum_{k\in [B]}\bm X_k\right]=Bp=T^{3/2}\cdot T^{-1/2}=T$, \eqref{thmgaplow_eq5} holds by the Hoeffding's inequality \cite{hoeffding1963probability} and \eqref{thmgaplow_eq6} follows from $t=\sqrt{T}$ and $B=T^{3/2}$.
By combining \eqref{thmgaplow_eq2} and \eqref{thmgaplow_eq6} we obtain
\begin{equation}\label{thmgaplow_eq7}
\frac{\sigma(\pi')}{\sigma(\pi)}\geq \frac{\left(1-\frac{1}{e^{2/\sqrt{T}}}\right)\cdot (T-\sqrt{T})}{ \left(1-\left(1-\frac{1}{\sqrt{T}}\right)^{\sqrt{T}}\right)\cdot T}.
\end{equation}
As the right-hand part of \eqref{thmgaplow_eq7} tends to $\frac{1}{1-\frac{1}{e}}=\frac{e}{e-1}$ for $T$ tending to $\infty$, by taking a sufficiently large perfect square $T:=T(\epsilon)$ in instance $I(\epsilon)$ we obtain
\begin{equation*}
\frac{\sigma(\pi')}{\sigma(\pi)}\geq \frac{\left(1-\frac{1}{e^{2/\sqrt{T}}}\right)\cdot (T-\sqrt{T})}{ \left(1-\left(1-\frac{1}{\sqrt{T}}\right)^{\sqrt{T}}\right)\cdot T}> \frac{e}{e-1}-\epsilon.
\end{equation*} By the initial observations, the above inequality implies that the budget adaptivity gap of $I(\epsilon)$ is higher than $\frac{e}{e-1}-\epsilon$, and by the arbitrariness of $\epsilon$ the claim follows. 
\end{proof}
%
%

\section{Conclusions}
In this paper we presented the first results for the problem of adaptively optimizing the value of a stochastic submodular function over multiple rounds with total budget constraints, when the
state distributions are known. Our analysis aims to provide theoretical guarantees on the approximation
that can be achieved in polynomial time.
We proposed two algorithms: one based on a dynamic programming approach, that is able to compute the optimal multi-round adaptive policy in polynomial time whenever a single-round optimization problem associated with each round can be efficiently solved (e.g., if the number of items is constant and the number of random realizations is polynomial), and a simpler greedy algorithm that, despite being only partially adaptive, still returns a $1/2(1-1/e-\epsilon)$-approximation of the optimum. Furthermore, by introducing the budget-adaptivity gap, we measured how much is lost, in terms of value obtained, if we distribute the budget over all rounds in advance, i.e. in a non-adaptive way. 

A first open problem left by our work is that of closing the gap on the approximation achievable in polynomial time, that is at least $1/2(1-1/e-\epsilon)$ by our results and at most $1-1/e$ because of the hardness result of \cite{Feige98} (holding for simpler non-adaptive single-round maximization problems). Furthermore, another open problem is that of finding tight bounds on the budget-adaptivity gap, which we have shown to be between $e/(e-1)$ and $2$. 

An interesting open question is related to the efficiency of the above greedy algorithm in terms of run-time.  Indeed, even if the proposed approximation algorithm we provided is essentially simple and it is proved to run in polynomial time, it involve some routines, such as Monte Carlo simulations, that are known to limit their practical application. Anyway, we believe that, as occurred
for many similar problems, our algorithm
may foster further research towards more practical approaches, that, e.g., replace the more expensive routines
with heuristics. Not only, but our algorithm will also serve as a benchmark for the evaluation of these heuristics.

A further research direction is that of combining multiple rounds, budget and adaptivity in other classical optimization problems, as we did in submodular maximization problems, and then study computational complexity, approximation and budget-adaptivity gap. 



Finally, we recall that the algorithms designed in this work assume that state distributions are known in advance, even if their realization is unknown. It would be undoubtedly interesting to explore the existence of algorithms with good approximation guarantees even if the state distributions must be learnt or they are provided by a some learning framework (without any precision guarantee).

\newpage
\bibliographystyle{plain}
\bibliography{aaai24}
\appendix
\section{Mapping between SBMSm and Specific Multi-round Problems}\label{apppendix:applications}
\subsection{Multi-round Influence Maximization}
We start by showing how an instance of the MR-IM-IC problem can be mapped to an instance of SBSMm.
Let $M = (T, B, V, E, (p_t)_{t\in [T]}, (w_t)_{t \in [T]})$ be an instance of MR-IM-IC and let $I=(T, B, V, H, (\mathcal{P}_t)_{t\in [T]}, (f_t)_{t \in [T]})$ denote the generic instance of SBSMm.
The mapping works as follows:
\begin{itemize}
 \item $T$ and $B$ are exactly the same as in $M$ (hence, we use exactly the same name).
 \item The set $V$ of items is the set of nodes in the social network.
 \item The set $H(v)$ of possible local states of $v$ corresponds to the possible sets of nodes influenced by an independent cascade starting from node $v$, i.e., $H(v) = \{X \subseteq V \mid x\text{ is reachable from }v\text{ in }G,\ \forall x\in X\}$; we observe that $|H(v)|$ is in general exponential in the input size, despite it can be succinctly represented in polynomial space.
 \item For any global state $\eta_t\in H$ and $v\in V$, $\mathcal{P}_t(\eta_t(v))$ is the probability that an independent cascade influences exactly the $\eta_t(v) \in H(v)$ nodes when started from $v$. We can succinctly define $\mathcal{P}_t$ through the concept of {\em live edges} \cite{kempe2003maximizing}. The {\em live-edge graph} $ G^L_t=(V, E^L)$ at round $t$ is a random subgraph of $G$ obtained by including each edge $(u,v)\in E$ in $ E^L_t$ with probability $p_{t,(u,v)}$, independently of the random inclusion of the other edges; each edge $(u,v)\in G_t^L$ is called {\em live} as it can be used to spread the influence from $u$ to $v$, while each edge $(u,v)\notin G_t^L$ is called {\em dead} as it cannot do it. 
We observe that the probability of a given subgraph $G=(V,E')$ to be the live-edge graph is $\mathbb{P}[ G^L_t=G']=\prod_{(u,v)\in E'}p_{t,(u,v)}\prod_{(u,v)\in E\setminus E'}(1-p_{t,(u,v)})$. Then, $\mathcal{P}_t(\eta_t(v))$ can be equivalently defined as the sum of $\mathbb{P}[ G^L_t=G']$, over all subgraphs $G'=(V,E')$ of $G$ such that the set of nodes reachable from $v$ in $G'$ coincides with $\eta_t(v)$, i.e., $\mathcal{P}_t(\eta_t(v))$ is the probability that the set of nodes reachable from $v$ in $G_t^L$ is exactly $\eta_t(v)$. 
%
%
%
 \item The objective function $f_t$ is defined as $f_t(S, \eta_t) = \sum_{v\in \bigcup_{u\in S}\eta_t(u)} w_t(v)$.
\end{itemize}
The mapping between a multi-round policy $\pi$ for IM-IC and a multi-round policy $\pi'$ for SBMSm is determined as follows. The multi-round policy $\pi$, at a certain step of round $t$, either moves to the subsequent round, or adds a new seed $v\in V$ to the previously selected seeds set $S^t$  and observes the newly activated nodes. For the sake of simplicity, we will assume that $\pi$, after selecting an item $v$ as a seed, is even able to observe the entire set of nodes that $v$ would have activate in graph $G$ if no seed had been previously selected. This assumption is not restrictive, as knowing if a seed $v$ would have or not influence some already active nodes does not give useful information to the policy, and will be used below only to define the mapping with SBMSm. The corresponding multi-round policy $\pi'$ of SBMSm, at a certain step of round $t$, either moves to the subsequent round, or adds a new seed $v\in V$ to the previously selected items $S^t$ and observes the local state $\eta_t(v)$. Considering that $\eta_t(v)$ represents the set of nodes influenced by a cascade starting from node $v$, the mapping between $\pi$ and $\pi'$ easily follows.

According to the mapping between $M$ and $I$ and that between the multi-round policies of such instances, it is immediate to check that $\varsigma(\pi) = \sigma(\pi')$, and thus finding the policy maximizing the latter is equivalent to find a policy maximizing the former.

\paragraph{Properties of $f_t$.}
In order to establish that $I$ is a valid instance for SBMSm we need also to prove that $f_t$ is polynomially computable and representable, it is polynomially bounded, and it satisfies monotonicity and  adaptive submodularity. Polynomial representation and computation immediately follows from the definition above. As for polynomial boundedness, this property is  satisfied if there exist constants $c,C>0$ such that $w_t(v) \leq C$ for every $v$ and every $t$, and there exist a round $t$ and a node $v$ such that $c\leq w_t(v)$; we observe that for standard unweighted influence maximization problems, where the objective function simply counts the number of influenced nodes, we can set $C=c=1$. Given the constants $c,C$ defined above, the values $\Lambda :=Cn$ and $\lambda:=c$ immediately satisfy the polynomial boundedness. 

The properties of monotonicity and adaptive submodularity have been formally shown for the standard single-round adaptive influence maximization problem by \cite{golovin2011adaptive}, and clearly extend to the multi-round setting considered here. For the sake of completeness, we provide the main intuitions on how these properties can be checked. Monotonicity is immediately satisfied since the information diffusion does not decrease if the set of seeds expands. As for the adaptive submodularity, it is not hard to check that this condition is satisfied: intuitively, for each realization of the live-edge graph, as described above, the increment of $f_t$ due to the selection of a new seed $v$ is given by nodes reachable from $v$ in this live graph, but not reachable by previously selected seeds, and it is clearly impossible to expand this subset with more previously selected seeds.

To show strong adaptive submodularity, we focus only on property (ii), as property (i) coincides with standard adaptive submodularity. We observe that the maximum expected increment of the value obtainable through the adaptive selection of \( k \) items, after observing a given partial state, is equivalent to the optimal adaptive value \( OPT(G') \) achievable in the subgraph \( G' \) obtained from \( G \) by removing the previously influenced nodes and their adjacent edges. Thus, the optimal adaptive value \( OPT(G) \) achievable in the whole graph satisfies \( OPT(G) \geq OPT(G') \), and by the above observation, property (ii) of strong adaptive submodularity holds.

\subsection{Multi-round Stochastic Probing}
Let $M = (T, B, V, (p_t)_{t\in [T]},(g_t)_{t\in [T]})$ be an instance of the MR-BSS problem, and let us now show how this instance can be modelled as an SBMSm instance $I=(T, B, V, H, (\mathcal{P}_t)_{t\in [T]}, (f_t)_{t \in T})$: $T$, $B$, and $V$ are exactly the same as in $M$; $H(v) = \{0, 1\}$, with $0$ meaning that item $v$ is inactive, and $1$ meaning that it is active, and thus $H = \{0, 1\}^n$; $\mathcal{P}_t(v) = p_t(v)$; and $f_t(S, \eta_t)=g_t(\{v\in S \colon \eta_t(v)=1\})$, i.e., we evaluates among the set $S$ of probed items, only the ones that happens of being active.

According to this mapping between $M$ and $I$, we have an obvious mapping between a multi-round policy $\pi$ for MR-BSS and a multi-round policy $\pi'$ for SBMSm, thus finding the policy maximizing the latter is equivalent to find a policy maximizing the former.

\paragraph{Properties of $f_t$.}
The polynomial boundedness of $f_t$ holds by using the same values of $\lambda$ and $\Lambda$ defined for $g_t$. The monotonicity and the adaptive submodularity of $f_t$ immediately follow from the monotonicity and the submodularity of $g_t$, respectively.

As for property (ii) of strong adaptive submodularity, let $k\geq 1$ be an integer, $S_t$ be a given subset of items and $\tilde{\bm S}_t$ be an adaptive selection of at most $k+|S_t|$ items constrained by the observation of the binary states of items in set $S_t$. Let $g_t$ be the submodular function that characterizes $f_t$ in the considered instance of MR-BSS. Considering that the binary states $(\bm \phi(v))_{v\in V}$ are independent, we can run the same adaptive selection as $\tilde{\bm S}_t\setminus S_t$, in a sub-instance obtained by restricting the set of items to $V':=V\setminus S_t$ and by taking as new function $g'_t$ that defined as 
$
g'_t(U)=g_t(U\cup S_t)-g_t(S_t)
$. By the above definition, we have that the value of $g'_t(U)$ under the adaptive selection $\tilde{\bm S}_t\setminus S_t$ in the considered sub-instance is equal to the expected increment of $g_t$ under the same adaptive selection in the initial instance,  after (selecting and) observing $S_t$. Furthermore, since inequality $g_t(U\cup S_t)-g_t(S_t)\leq g_t(U)$ holds (by the submodularity of $g_t$), the above expected increment is at most equal to the expected value of $g_t$ achieved by the selection of $\tilde{\bm S}_t\setminus S_t$. Thus, as the number of items selected in $\tilde{\bm S}_t\setminus S_t$ is at most $k$, we have that the expected value of an optimal adaptive policy that selects $k$ items at round $t$ is at least equal to the expected increment of $g_t$, due to the selection of items in $\tilde{\bm S}_t\setminus S_t$ after (selecting and) observing $S_t$. We conclude that property (ii) of strong adaptive submodularity follows. 
\section{A Tree-based Algorithm for the Single-Round Problem $\SRM$}\label{sec_fpt}
\subsection{The Single-round Optimization Problem as a Game.}To design a single-round policy $\pi_t$ that computes $R(t,b)$ for any round $t$ and budget $b$, it is useful to describe the $\SRM$ problem at round $t$ in terms of a suitable two-player extensive form game between two players \cite{Kuhn2003}: (i) the policy $\pi_t$, that decides, at each step, either to interrupt the game and collecting a reward or to continue by selecting another item; (ii) the nature $\mathcal{N}$, that reveals a partial state after each item selection. 

Specifically, the game is modeled as a directed rooted tree $\mathcal{T}_t$ of depth $2n+1$. The policy $\pi_t$ will play at all nodes at even depth, whereas the nature $\mathcal{N}$ plays at all nodes at odd depth. We will denote with $\mathcal{T}_t^o$ the set of nodes in which $\pi_t$ plays, and with $\mathcal{T}_t^e$ the set of nodes in which ${\cal N}$ plays. Each leaf of the tree corresponds to an outcome of the game, in which the policy $\pi_t$ collects a reward depending on the reached leaf. 

In the following, we define the tree $\mathcal{T}_t$, so as the game between $\pi_t$ and ${\cal N}$, starting from the nodes at depth $d\in\{0,1\}$, then at depth $d\in\{2,3\}$, and iteratively at depth $d\in\{2(i-1),2i-1\}$, for any $i\in [n+1]$.

{\bf Base Case ($d\in \{0,1\}$):} The root $u^*=u_0$ of $\mathcal{T}_t$, that belongs to $\mathcal{T}_t^o$ (as its depth is $0$), is the initial configuration of the game. The root has $n+1$ outgoing edges, where each edge $(u_0,u_1)$ has a label $a=L(u_0,u_1)$ corresponding to the action $a$ that $\pi_t$ can perform at step $1$. In particular, either the action/label $a$ is a string STOP that will denote the termination of the single-round policy, or $a$ is an item $v_1 \in V$ that $\pi_t$ can select at step~$1$. Given the action $a$ chosen by $\pi_t$ at step $1$, the new configuration of the game is obtained by following the corresponding labeled edge (i.e., the edge $(u_0,u_1)$ such that $a=L(u_0,u_1)$), and it corresponds to a node $u_1$ at depth $1$. If $a=$STOP then the reached node $u_1$ is a leaf and the game ends. In such a case, the reward of $\pi_t$ is $R(t+1,b)$, as $\pi_t$ did not select any node and the whole budget $b$ can be used in rounds $t+1,\ldots, T$. Otherwise, the reached node $u_1$ has as many outgoing edges as partial states $\langle \{v_1\}\rangle$ that can be observed after selecting the item $v_1$, and each outgoing edge $(u_1,u_2)$ is labeled with a distinct partial state. At this point the nature $\mathcal{N}$ picks a partial state $\langle \{v_1\}\rangle$, drawn according to the probability distribution induced by $\mathcal{P}_t$ on partial states of type $\langle \{v_1\}\rangle$ (i.e., observable after selecting $v_1$).  Then, the new configuration of the game is obtained by following the outgoing edge labeled with the picked partial state (i.e., the edge $(u_1,u_2)$ such that $\langle \{v_1\}\rangle=L(u_1,u_2)$), and it corresponds to a node $u_2$ at depth $2$, that belongs to $\mathcal{T}_t^o$.

{\bf General Case ($d\in \{2(i-1),2i-1\}$):} Iteratively, for any $i\in [n+1]$, each node $u_d\in \mathcal{T}_t$ at depth $d=2(i-1)$ corresponds to a possible outcome of policy $\pi_t$ obtained at the end of step $i-1$, where $i-1$ items $v_1,\ldots, v_{i-1}$ have been already selected and a partial state $\langle \{v_1,\ldots, v_{i-1}\}\rangle$ has been observed. $u_d$ has $n-i+2$ outgoing edges $(u_d,u_{d+1})$, each one labeled with the action $a=L(u_d,u_{d+1})$ that $\pi_t$ can performs at step $i$: either $a$=STOP, or $a$ is equal to an item $v_i\in V\setminus\{v_1,\ldots, v_{i-1}\}$ that $\pi_t$ can select at step $i$; we observe that, for $i=n+1$, there is a unique outgoing edge labeled with action $a=$STOP. Again, given the action $a$ chosen by $\pi_t$ at step $i$, the new configuration of the game is obtained by  following the corresponding labeled edge and reaching a node $u_{d+1}$ at depth $d+1$. If $a=$STOP, then $u_{d+1}$ is a leaf, the game ends and the profit for $\pi_t$ is the expected value of $f_t$ constrained by the observation of partial state $\langle \{v_1,\ldots, v_{i-1}\}\rangle$, plus the optimal expected value $R(t+1,b-i+1)$ achievable in the subsequent rounds by using the residual budget (i.e., $b-i+1$ units).  Otherwise, $u_{d+1}$ has as many outgoing edges as partial states $\langle \{v_1,\ldots, v_{i}\}\rangle$ that can be observed after selecting $v_i$ (constrained by the previous observation of $\langle \{v_1,\ldots, v_{i-1}\}\rangle$), and each outgoing edge is labeled with a distinct partial state. At this point the nature $\mathcal{N}$ picks a partial state $\langle \{v_1,\ldots, v_{i}\}\rangle$ drawn according to the probability distribution induced by $\mathcal{P}_t$ on partial states of type $\langle \{v_1,\ldots,v_i\}\rangle$, conditioned by the observation of $\langle \{v_1,\ldots, v_{i-1}\}\rangle$. The new configuration of the game is obtained by following the outgoing edge labeled with the picked partial state, and it corresponds to a node $u_{d+1}$ at depth $d+1$.

\subsection{From the Game to the Optimal Solution. } Given a node $u\in \mathcal{T}_t^o$, let $\langle S(u)\rangle$ denote the partial state observed when the configuration of the game is determined by $u$. We have that the above surjective map between nodes of $\mathcal{T}_t^o$ and partial states defines an equivalence between the tree-based game and the problem of finding a policy $\pi_t$ that maximizes \eqref{optimum_singleround}. Indeed, let $a^*(u)$ denote the action that maximizes the expected value of the reward collected by $\pi_t$, assuming that the actual configuration of the game is $u$, and let $R(u)$ denote the resulting optimal expected value.  We have that $a^*(u)$ is also the action that $\pi_t$ must follow for any partial state $\langle S(u)\rangle$ in order to maximize \eqref{optimum_singleround}. Then, denoting the root of $\mathcal{T}_t$ by $u^*$, we have that $R(u^*)$ coincides with $R(t,b)$, that is the optimal value of the $\SRM$ problem. 

We observe that, by following standard approaches to find the optimal strategy of a player in extensive-form games \cite{Kuhn2003}, $a^*(u)$ and $R(u)$ can be computed for any $u\in \mathcal{T}_t^o$ by recursively exploring the tree $\mathcal{T}$. By the equivalence between the tree-based game and the single-round maximization problem, the optimal strategy for the former can be easily turned into an optimal solution for the latter.

\subsection{When the Single-round Problem can be Efficiently Solved?}
The tree-based algorithmic procedure to compute an optimal policy for $\SRM$ can be run in linear time w.r.t. the size of the game tree $\mathcal{T}_t$, as it simply requires to recursively visit all the nodes in tree. Thus, $\SRM$ can be solved efficiently whenever the game tree ${\cal T}_t$ contains a number of nodes that is polynomial in the size of the input instance. To this aim, in the proposition below, we relate the size of ${\cal T}_t$ to the number of items and states, and we next show that there may be practical settings in which this is not rare. 
\begin{proposition}\label{prop_tree_nodes}
${\cal T}_t$ has $O(2^n n!|H|)$ nodes. 
\end{proposition}
\begin{proof}
Let $z$ denote the number of leaves of ${\cal T}_t$, and let $|{\cal T}_t|$ denote the number of nodes of ${\cal T}_t$. We first show that 
\begin{equation}\label{eq_tree_nodes1}
|{\cal T}_t|\in O(z).
\end{equation}
By exploiting the structure of $\mathcal{T}_t$, we have that each non-leaf node $u\in \mathcal{T}_t$ has exactly one child node that is a leaf. Then, denoting as $y$ the number of non-leaf nodes of $\mathcal{T}_t$, we have $y\leq z$. Thus, we have $|{\cal T}_t|=y+z\leq 2z$, and then \eqref{eq_tree_nodes1} follows.

Now, we show that 
\begin{equation}\label{eq_tree_nodes2}
z\in O(2^n n! |H|).
\end{equation}
Let $G$ be the function that maps each sequence $s$ of distinct items in $V$ and partial state state $\eta_t\in H$, to the leaf $u=G(s,\eta_t)$ of $\mathcal{T}_t$ associated with the outcome of $\pi_t$ obtained when the realized state is $\eta_t$ and $s$ is the sequence of items selected by $\pi_t$ at round $t$. We observe that, by the structure of $\mathcal{T}_t$, function $G$ is surjective. Thus, as the number of sequences of distinct items in $V$ is at most $2^n n!$ (at most $n!$ orderings for each of the $2^n$ subsets of $V$), we have that the number of leaves is at most $2^n n! |H|$, and this shows \eqref{eq_tree_nodes2}. 

By putting \eqref{eq_tree_nodes1} and \eqref{eq_tree_nodes2} together, the claim of the proposition follows. 
\end{proof}
Since the tree-based algorithm for $\SRM$ can be run in linear time w.r.t. to the size of $\mathcal{T}_t$ and, by Proposition \ref{prop_tree_nodes}, such size is $O(2^n n!|H|)$, we then have the following corollary:
\begin{corollary}\label{coro_opt}
$\SRM$ can be solved in $O(2^n n!|H|)$ time. 
\end{corollary}
Furthermore, by observing that the dynamic programming algorithm for SBMSm executes $\SRM$ once for each pair round-budget $(t,b)\in [T]\times \llbracket 0,B\rrbracket$, i.e., $O(BT)\subseteq O(nT^2)$ times, we have the following further corollary:
\begin{corollary}\label{coro_opt_dyn}
SBMSm can be solved in $O(2^n n!nT^2|H|)$ time.
\end{corollary}
If $n!$ and $|H|$ are polynomial in the input instance of SBMSm, then $2^n n!|H|$ is polynomial, too. Thus, by Corollaries \ref{coro_opt} and \ref{coro_opt_dyn}, both $\SRM$ and SBMSm can be run in  polynomial time in such case. 
Note that $n!$ is polynomial if, for instance, the number $n$ of items is very small, e.g.,  constant. Such a limitation is not so rare in several scenarios of interest, as it often occurs that in large populations only few items are really relevant (e.g., among the many friends on social networks only few of them are able to affect our opinions, or among the many potential job candidates only few effectively satisfy all the requirements at which one may be interested in). The limitation on the number of states $|H|$ occurs in many settings where the number of observable events is not so high. For instance, it could be the case of certain influence maximization problems where edge probabilities are correlated, and depends on a few possible stochastic events. Furthermore, it often occurs that, if the number of items is constant, the number of states $|H|$ is constant, too. For instance, in both influence maximization and stochastic probing, we have that the number of states, even if it grows exponentially in the number of items, it depends on such quantity only. Therefore, in such a case, if the number of items is constant, the number of states will be constant, too. 

Finally, by exploiting again the complexity bounds provided in Corollaries \ref{coro_opt} and \ref{coro_opt_dyn}, we have the following result on fixed-parameter tractability:
\begin{corollary}\label{corollary_fpt}
Both $\SRM$ and SBMSm are fixed-parameter tractable problems w.r.t. to $n$, if $|H|$ is polynomially bounded in the size of the input instance.
\end{corollary}
\section{Missing Proofs of Sections \ref{sec:sampling}}\label{appendix:oracle}
\subsection{Proof of Proposition \ref{prop:oracle_property}}
Let $t\in [T]$, $\langle S\rangle$ be a partial state at round $t$ and, for any $v\in V$, let $\bm X_{t,v,\langle S\rangle}$ and $X^1_{t,v,\langle S\rangle},\ldots, X^q_{t,v,\langle S\rangle}$ be the random quantities considered in the definition of {\sf Oracle1}. As $\bm X_{t,v,\langle S\rangle}$ has values in interval $[0,\Lambda]$ (by the polynomial boundedness of the objective function), we have that
\begin{align*}
\mathbb{P}\left[\exists v\in V:\Bigg|\sum_{j=1}^q \frac{X^j_{t,v,\langle S\rangle}}{q}-\mathbb{E}\left[\bm X_{t,v,\langle S\rangle}\right]\Bigg|\geq \delta/2\right]
&\leq \sum_{v\in V}\mathbb{P}\left[\Bigg|\sum_{j=1}^q \frac{X^j_{t,v,\langle S\rangle}}{q}-\mathbb{E}\left[\bm X_{t,v,\langle S\rangle}\right]\Bigg|\geq \delta/2\right]\\
&\leq n\cdot 2 \exp\left(-\frac{2(\delta/2)^2q}{\Lambda^2}\right)\\
&\leq n\cdot 2 \exp\left(-\frac{2(\delta/2)^2q(\delta,\xi)}{\Lambda^2}\right)\\
&=\xi.
\end{align*}
where the second inequality follows by the Chernoff-Hoeffding's inequality, and the last equality holds by construction of $q(\delta,\xi)$. The above inequalities imply that
\begin{equation}
\label{model_hoeff1}
\begin{aligned}
\mathbb{P}\left[|\Delta_t(v|\langle S_t\rangle)-\tilde{\Delta_t}(v|\langle S_t\rangle)|\leq \delta/2\ \forall v\in V\right]
&=\mathbb{P}\left[\Bigg|\mathbb{E}\left[\bm X_{t,v,\langle S\rangle}\right]-\sum_{j=1}^q \frac{X^j_{t,v,\langle S\rangle}}{q}\Bigg|\leq \delta/2\ \forall v\in V\right]\\
&\geq 1-\xi.
\end{aligned}
\end{equation}
Now, fix $v^*\in \arg\max_{v\in V}\Delta_t(v|\langle S_t\rangle)$ and $\tilde{v}\in \arg\max_{v\in V}\tilde{\Delta}_t(v|\langle S_t\rangle)$. For any $q\geq q(\delta,\xi)$, we have that
\begin{align}
\mathbb{P}\left[\max_{v\in V}\Delta_t(v|\langle S_t\rangle)-{\Delta_t}(\tilde{v}|\langle S_t\rangle)\leq \delta\right]
&=\mathbb{P}\left[\Delta_t(v^*|\langle S_t\rangle)-{\Delta_t}(\tilde{v}|\langle S_t\rangle)\leq \delta\right]\nonumber\\
&\geq \mathbb{P}\left[(\Delta_t(v^*|\langle S_t\rangle)-\tilde{\Delta_t}(v^*|\langle S_t\rangle))+(\tilde{\Delta_t}(\tilde{v}|\langle S_t\rangle)-{\Delta_t}(\tilde{v}|\langle S_t\rangle))\leq \delta\right]\label{oracle.eq1}\\
&\geq \mathbb{P}\left[|\Delta_t(v|\langle S_t\rangle)-\tilde{\Delta_t}(v|\langle S_t\rangle)|\leq \delta/2\ \forall v\in \{\tilde{v},v^*\}\right]\nonumber\\
&= \mathbb{P}\left[\left|\mathbb{E}\left[\bm X_{t,v,\langle S_t\rangle}\right]-\sum_{j=1}^q \frac{X^j_{t,v,\langle S\rangle}}{q}\right|\leq \delta/2\ \forall v\in \{\tilde{v},v^*\}\right]\label{oracle.eq1b}\\
&\geq \mathbb{P}\left[\Bigg|\mathbb{E}\left[\bm X_{t,v,\langle S_t\rangle}\right]-\sum_{j=1}^q \frac{X^j_{t,v,\langle S_t\rangle}}{q}\Bigg|\leq \delta/2\ \forall v\in V\right]\nonumber\\
&\geq 1-\xi,\label{oracle.eq2}
\end{align}
where \eqref{oracle.eq1} holds as $-\tilde{\Delta_t}(v^*|\langle S_t\rangle)+\tilde{\Delta_t}(\tilde{v}|\langle S_t\rangle)\geq 0$ (by the optimality of $\tilde{v}$ for $\tilde{\Delta_t}(\tilde{v}|\langle S_t\rangle)$), \eqref{oracle.eq1b}  holds since $\Delta_t(v|\langle S_t\rangle)=\mathbb{E}\left[\bm X_{t,v,\langle S_t\rangle}\right]$ and $\tilde{\Delta_t}(v|\langle S_t\rangle)=\sum_{j=1}^q \frac{X^j_{t,v,\langle S\rangle}}{q}$, and \eqref{oracle.eq2} follows from \eqref{model_hoeff1}. Finally, Property (i) follows from \eqref{model_hoeff1}.

Property (ii) immediately follows from: (a) the polynomial boundedness property, that polynomially relates the quantity $\Lambda$ (appearing in $q(\delta,\xi)$) to the size of the input instance; (b) the existence of an oracle that picks samples drawn from the underlying probability distributions (whose existence has been implicitly assumed in the model definition of Section \ref{sec:model}). 
\subsection{Proof of Proposition \ref{prop:oracle2_property}}
Given $t\in [T]$, $i\in [n]$ and $q\geq 1$, let $\bm Y_{t,i}$ and $Y_{t,i}^1,\ldots, Y_{t,i}^q$ be the random quantities considered in the definition of {\sf Oracle2}. As $\bm Y_{t,i}$ has values in $[0,\Lambda]$ we have that:
\begin{align*}
\mathbb{P}\left[\exists t\in [T],i\in [n]:\Bigg|\sum_{j=1}^q \frac{Y^j_{t,i}}{q}-\mathbb{E}\left[\bm Y_{t,i}\right]\Bigg|\geq \delta\right]
&\leq \sum_{t\in [T]}\sum_{i\in [n]}\mathbb{P}\left[\Bigg|\sum_{j=1}^q \frac{X^j_{t,v,\langle S\rangle}}{q}-\mathbb{E}\left[\bm X_{t,v,\langle S\rangle}\right]\Bigg|\geq \delta\right]\\
&\leq T\cdot n\cdot 2 \exp\left(-\frac{2\delta^2q}{\Lambda^2}\right)\\
&\leq T\cdot n\cdot 2 \exp\left(-\frac{2\delta^2q'(\delta,\xi)}{\Lambda^2}\right)\\
&=\xi.
\end{align*}
where the second inequality holds by the Chernoff-Hoeffding's inequality, and the last equality holds by construction of $q'(\delta,\xi)$. The above inequalities imply that
\begin{align*}
\mathbb{P}\left[|\tilde{\Delta}'_{t,i}-\Delta'_{t,i}|\leq \delta,\ \forall t\in [T],i\in [n]\right]
&=\mathbb{P}\left[\Bigg|\sum_{j=1}^q \frac{Y^j_{t,i}}{q}-\mathbb{E}\left[\bm Y_{t,i}\right]\Bigg|\leq\delta,\ \forall t\in [T],i\in [n]\right]\\
&\geq 1-\xi,
\end{align*}
and this shows property (i).

Property (ii) holds by similar arguments as in Proposition~\ref{prop:oracle_property}.
\section{Missing Proofs from Section \ref{subsec:approx}}\label{appendix:thm1}
\subsection{Full Proof of Lemma \ref{lem1}}
Continuing from the proof sketch given in Section \ref{subsec:approx}, it remains to analyze the case in which $d\geq 0$ is a generic real number, that extends in a non-trivial way the proof arguments given in the proof sketch. 

Let $d\in \mathbb{R}_{\geq 0}\setminus \mathbb{Z}$, $\pi^*_t(d)$ be an arbitrary optimal adaptive policy at round $t$ that selects $d$ items in expectation and $p:=d-\lfloor d\rfloor\in (0,1)$.
For any integer $i\geq 1$, let $\bm b_i$ be an independent random variable that is equal to $\lceil d\rceil$ with probability $p$ and equal to $\lfloor d\rfloor$ with probability $1-p$, let ${\bm S_i'}$ denote the (possibly empty) random set of at most $\bm b_i$ items selected by $\pi^*_t(d)$ after having already selected $\sum_{j=1}^{i-1}\bm b_j$ items (that is, $\sqcup_{i\in \mathbb{N}}\bm S_i'$ is the set of items selected by  $\pi^*_t(d)$ at round $t$), and let $\vec{\bm b}_{i-1}$ denote the vector $(\bm b_1,\ldots, \bm b_{i-1})$. Given $i\geq 2$ and a vector $\vec{b}_{i-1}\in \{\lfloor d\rfloor,\lceil d\rceil\}^{i-1}$, let $A(\vec{b}_{i-1})$ denote the event ``at least $\sum_{j=1}^{i-2}b_j+\lceil d\rceil$ items have been selected by $\pi^*_t(d)$ and $\vec{\bm b}_{i-1}=\vec{b}_{i-1}$'' and, given $b_i\in \{\lfloor d\rfloor,\lceil d\rceil\}$, let $\Delta(\pi^*_t(d)|A_i(\vec{b}_{i-1}),b_i)$ denote the expected increment of $f_t$ gained from the selection of ${\bm S_i'}$ (after the set $\sqcup_{j=1}^{i-1}\bm S_j'$ has already been selected) and conditioned by events $A_i(\vec{b}_{i-1})$ and ``$\bm b_i=b_i$''.

We have\footnote{To shorten the notation in the summations, we implicitly assume that $\vec{b}_{i-1}$ belongs to $\{\lfloor d\rfloor,\lceil d\rceil\}^{i-1}$ and $b_i$ belongs to $\{\lfloor d\rfloor,\lceil d\rceil\}$.}:
 \begin{align}
 &\overline{OPT}_t(d)=\sigma(\pi^*_t(d))\nonumber\\
 &=\sum_{i=1}^{\infty} \sum_{\vec{b}_{i-1}}\sum_{b_i}\mathbb{P}[A(\vec{b}_{i-1})\wedge \bm b_i=b_i]\cdot \Delta(\pi^*_t(d)|A_i(\vec{b}_{i-1}),b_{i})\label{app:lem1:ineq00}\\
  &=\sum_{i=1}^{\infty} \sum_{\vec{b}_{i-1}}\sum_{b_i}\mathbb{P}[A(\vec{b}_{i-1})]\cdot \mathbb{P}[\bm b_i=b_i]\cdot \Delta(\pi^*_t(d)|A_i(\vec{b}_{i-1}),b_{i})\label{app:lem1:ineq01}\\
  &=\sum_{i=1}^{\infty} \sum_{\vec{b}_{i-1}}\mathbb{P}[A(\vec{b}_{i-1})]\cdot \sum_{b_i}\mathbb{P}[\bm b_i=b_i]\cdot \Delta(\pi^*_t(d)|A_i(\vec{b}_{i-1}),b_{i})\nonumber\\
   &\leq \sum_{i=1}^{\infty} \sum_{\vec{b}_{i-1}}\mathbb{P}[A(\vec{b}_{i-1})]\cdot \sum_{b_i}\mathbb{P}[\bm b_i=b_i]\cdot OPT_t(b_{i})\label{app:lem1:ineq02}\\
 &= \sum_{i=1}^{\infty} \sum_{\vec{b}_{i-1}}\mathbb{P}[A(\vec{b}_{i-1})]\cdot(p\cdot OPT_t(\lceil d\rceil)+(1-p)\cdot OPT_t(\lfloor d\rfloor))\nonumber\\
 &=\sum_{i=1}^{\infty} \sum_{\vec{b}_{i-1}}\mathbb{P}[A(\vec{b}_{i-1})]\cdot OPT_t(d),\label{app:lem1:ineq0}
 \end{align}
 where:  \eqref{app:lem1:ineq00} is obtained by decomposing the total value of $\pi^*_t(d)$ in terms of the contribution given by the sets $\bm S'_i$, and by observing that $\bm S'_i>0$ holds only if $A_i(\vec{\bm b}_{i-1})$ is true;  \eqref{app:lem1:ineq01} holds since $\bm b_i$ does not depend on event $A(\vec{b}_{i-1})$;  \eqref{app:lem1:ineq02} holds by the strong adaptive submodularity; \eqref{app:lem1:ineq0} holds by definition of $OPT_t(d)$. 
 
 Furthermore, we have
 \begin{align}
 d&= \sum_{i=2}^{\infty}\mathbb{E}_{\bm \eta_t}[\#{\bm S_{i-1}'}]\nonumber\\
&\geq  \sum_{i=2}^{\infty} \sum_{\vec{b}_{i-1}}\mathbb{P}[A(\vec{b}_{i-1})]\cdot\mathbb{E}_{\bm \eta_t}[\#{\bm S_{i-1}'}|A(\vec{b}_{i-1})]\nonumber\\
&=  \sum_{i=2}^{\infty} \sum_{\vec{b}_{i-1}}\mathbb{P}[A(\vec{b}_{i-1})]\cdot\mathbb{E}_{\bm b_i}[\bm b_i]\label{app:lem1:ineq010}\\
&=  \sum_{i=2}^{\infty} \sum_{\vec{b}_{i-1}}\mathbb{P}[A(\vec{b}_{i-1})]\cdot d\label{app:lem1:ineq20}\\
 &\geq -d+\sum_{i=1}^\infty \sum_{\vec{b}_{i-1}}\mathbb{P}[A_{i}(\vec{b}_{i-1})]\cdot d\label{app:lem1:ineq3},
 \end{align}
where: \eqref{app:lem1:ineq010} holds since, if event $A(\vec{\bm b}_{i-1})$ is true, then at least $\lceil d\rceil$ items are selected in addition to $\sqcup_{j=1}^{i-2}\bm S'_j$ and, consequently, $\#\bm S'_i$ is necessarily equal to the value of random variable $\bm b_i$ (whose probability distribution is not conditioned by $A(\vec{\bm b}_{i-1})$); \eqref{app:lem1:ineq20} holds since $\mathbb{E}_{\bm b_{i}}[\bm b_{i}]=p \lceil d\rceil+(1-p)\lceil d\rceil=d$. By rearranging \eqref{app:lem1:ineq3}, we obtain 
 \begin{equation}\label{app:lem1:ineq3.1}
\sum_{i=1}^\infty \sum_{\vec{b}_{i-1}}\mathbb{P}[A(\vec{b}_{i-1})]\leq 2,
\end{equation}  
and by applying this inequality to \eqref{app:lem1:ineq0} we obtain 
\begin{equation*}
\overline{OPT}_t(d)\leq \sum_{i=1}^{\infty}\sum_{\vec{b}_{i-1}} \mathbb{P}[A(\vec{b}_{i-1})]\cdot OPT_t(d)\leq 2\cdot OPT_t(d),
\end{equation*}
that shows the claim.
\subsection{Full Proof of Theorem \ref{thm1} (for approximated oracles)}
To show the approximation guarantee of Theorem \ref{thm1} for the case of approximated oracles (i.e., with $\epsilon>0$, $\xi:=\delta:=\frac{\lambda c\epsilon}{B(4+3\Lambda)}$ and $0<c\leq 1$), we only need to modify partially the results in Parts 3-4 used for the case of exact oracles. In particular, $OPT$, $\vec{b}^*$ and $OPT_t(b)$ are defined as in the proof sketch of Theorem \ref{thm1}, while the definition of $GR_t(b)$, $\vec{b}$ and $GR_t(\vec{b})$ should be extended to the values $\delta,\xi>0$ adopted for approximated oracles. Inequalities (\ref{thm1:eq1}-\ref{thm1:eq4}) continue to hold in this setting. 

By Lemma \ref{lem5} below (that uses the adaptive submodularity), we show that the approximate greedy single-round policy with budget $b$ achieves an approximation of $(1-1/e))$, up to an addend $(\delta+\Lambda\xi)b$, and this leads to
\begin{equation}\label{app:thm1:eq5}
\begin{split}
&\sum_{t\in [T]}\left(1-\frac{1}{e}\right)OPT_t(b_t^*)-(\delta+\xi\Lambda)B=\sum_{t\in [T]}\left(\left(1-\frac{1}{e}\right)OPT_t(b_t^*)-(\delta+\xi\Lambda)b_t^*\right)\\
&\quad \leq \sum_{t\in [T]}{GR}_t(b^*_t)={GR}(\vec{b}^*).
\end{split}
\end{equation}
By exploiting the greedy assignment of ${\sf BudgetGr}_{\delta,\xi}$ and the adaptive submodularity, we have 
\begin{equation}\label{app:thm1:eq6}
GR(\vec{b}^*)-(3\delta+2\xi\Lambda)B\leq GR(\vec{b})
\end{equation}
(see Lemma \ref{lem6} below).
Furthermore, we have
\begin{equation}\label{app:thm1:eq7}
(4\delta+3\xi\Lambda)B\leq \epsilon\cdot \lambda\leq \epsilon\cdot OPT,
\end{equation}
where the first inequality holds by definition of $\delta$ and $\xi$, and the last one holds by Lemma \ref{lem7} below.
By combining inequalities (\ref{thm1:eq1}-\ref{thm1:eq4}) and (\ref{app:thm1:eq5}-\ref{app:thm1:eq7}) we obtain
\begin{align*}
&\frac{1}{2}\left(1-\frac{1}{e}-\epsilon\right)OPT\leq \left(1-\frac{1}{e}\right)\sum_{t\in [T]}OPT_t(b^*_t)-\epsilon\cdot OPT\leq \left(1-\frac{1}{e}\right)\sum_{t\in [T]}OPT_t(b^*_t)-(4\delta+3\xi\Lambda)B\\
& \quad \leq GR(\vec{b}^*)-(3\delta+2\xi\Lambda)B\leq GR(\vec{b}),
\end{align*}
and this shows the claim. 
\subsection{Technical Lemmas for the Full Proof of Theorem \ref{thm1} under Approximated Oracles}
\begin{lemma}\label{lem5}
For any $t\in [T]$ and integer $b\geq 0$, we have
$$
\left(1-\frac{1}{e}\right)OPT_t(b)-(\delta+\xi\Lambda)b\leq GR_t(b).
$$
\end{lemma}
\begin{proof}
Fix $t\in [T]$ and an integer $b\geq 0$. For any $i\in [b]$, let $GR^*_t(i-1,b)$ be the value of an adaptive policy $\pi^*_t(i-1,b)$ at round $t$ that first executes the greedy single-round policy ${\sf \SGR}_{\delta,\xi}$ at round $t$ with budget $i-1$, then runs an optimal adaptive policy $\pi^*_t$ at round $t$ with budget $b$, without observing the realizations coming from the greedy one, and finally selects the items previously chosen by both policies. By the monotonicity property, we have
\begin{equation}\label{lem5:ineq1}
OPT_t(b)\leq GR^*_t(i-1,b).
\end{equation}
Furthermore, by exploiting the adaptive submodularity, we can show that inequality
\begin{equation}\label{lem5:ineq2}
GR_t(i)\geq \frac{GR^*_t(i-1,b)}{b}+\left(1-\frac{1}{b}\right)GR_t(i-1)-(\delta+\xi\Lambda)
\end{equation}
holds for any $i\in [b]$. 
Indeed, let $\bm v_{t,i}$ denote the (random) item selected by the greedy single-round policy ${\sf \SGR}_{\delta,\xi}$ at round $t$ and step $i$, $\bm S_{t,i-1}:=\{\bm v_{t,1},\ldots, \bm v_{t,i-1}\}$ denote the set of the first $i-1$ items selected by ${\sf \SGR}_{\delta,\xi}$, and $\Delta_{t,i}:=GR_t(i)-GR_t(i-1)$ denote the expected increment of $f_t$ achieved at step $i$ of the greedy single-round policy. Furthermore, for any partial state $\langle S_{t,i-1}\rangle\sim \langle\bm S_{t,i-1}\rangle$, let $A(\langle S_{t,i-1}\rangle)$ denote the event ``$\max_{v\in V}\Delta_t(v|\langle  S_{t,i-1}\rangle)\leq \Delta_t(\bm v_{t,i}|\langle S_{t,i-1}\rangle)+\delta$''; by the property (i) of Proposition~\ref{prop:oracle_property}, we have that, for any $\langle S_{t,i-1}\rangle \sim \langle\bm S_{t,i-1}\rangle$, $A(\langle S_{t,i-1}\rangle )$ is true with probability at least $1-\xi$, and then, its complementary event $\overline{A}(\langle S_{t,i-1} \rangle )$ is false with probability at most $\xi$. Thus, we have
\begin{align}
&\mathbb{E}_{\langle \bm S_{t,i-1}\rangle}\left[\max_{v\in V}\Delta_t(v|\langle \bm S_{t,i-1}\rangle)\right]\nonumber\\
    &= \mathbb{E}_{\langle \bm S_{t,i-1}\rangle}\left[\mathbb{P}[A(\langle\bm S_{t,i-1}\rangle)]\cdot \mathbb{E}_{\bm v_{t,i}}\left[\max_{v\in V}\Delta_t(v|\langle \bm S_{t,i-1}\rangle)|A(\langle\bm S_{t,i-1}\rangle)\right]\right]\nonumber\\
    &\quad\quad + \mathbb{E}_{\langle \bm S_{t,i-1}\rangle}\left[\mathbb{P}[\overline{A}(\langle\bm S_{t,i-1}\rangle)]\cdot \mathbb{E}_{\bm v_{t,i}}\left[\max_{v\in V}\Delta_t(v|\langle \bm S_{t,i-1}\rangle)|\overline{A}(\langle\bm S_{t,i-1}\rangle)\right]\right]\nonumber\\
    &\leq  \mathbb{E}_{\langle \bm S_{t,i-1}\rangle}\left[\mathbb{P}[A(\langle\bm S_{t,i-1}\rangle)]\cdot \mathbb{E}_{\bm v_{t,i}}\left[\Delta_t(\bm v_{t,i}|\langle \bm S_{t,i-1}\rangle)+\delta|A(\langle\bm S_{t,i-1}\rangle)\right]\right]\nonumber\\
    &\quad\quad + \mathbb{E}_{\langle \bm S_{t,i-1}\rangle}\left[\mathbb{P}[\overline{A}(\langle\bm S_{t,i-1}\rangle)]\cdot \mathbb{E}_{\bm v_{t,i}}\left[\Lambda|\overline{A}(\langle \bm S_{t,i-1}\rangle)\right]\right]\label{lem5:ineq30}\\
    &\leq  \mathbb{E}_{\langle \bm S_{t,i-1}\rangle}\left[\mathbb{E}_{\bm v_{t,i}}\left[\Delta_t(\bm v_{t,i}|\langle \bm S_{t,i-1}\rangle)+\delta|\langle \bm S_{t,i-1}\rangle\right]\right]+ \mathbb{E}_{\langle \bm S_{t,i-1}\rangle}\left[\xi\cdot \Lambda\right]\nonumber\\
        &=  \mathbb{E}_{\langle \bm S_{t,i-1}\rangle}\left[\mathbb{E}_{\bm v_{t,i}}\left[\Delta_t(\bm v_{t,i}|\langle \bm S_{t,i-1}\rangle)\right]|\langle \bm S_{t,i-1}\rangle\right]+\delta+\xi\Lambda\nonumber\\
&=\Delta_{t,i}+\delta+\xi\Lambda,\label{lem5:ineq3}
\end{align}
where \eqref{lem5:ineq30} holds since $\Lambda$, by definition, is an upper bound on the maximum increment achievable from the selection of any item. Furthermore, for any $h\in [b]$, let $\bm v^*_{t,h}$ be the $h$-th (random) item selected by the optimal adaptive policy $\pi^*_t$, let $\bm S^*_{t,h-1}:=\{\bm v^*_{t,1},\ldots, \bm v^*_{t,h-1}\}$ denote the set of the first $h-1$ items selected by $\pi^*_t$ and let $\Delta^*_{t,h}$ be the expected increment of $f_t$ due to the addition of item $\bm v^*_{t,h}$ to set $\bm S_{t,i-1}\cup \bm S^*_{t,h-1}$. 
Given $h\in [b]$, we have
\begin{align}
&\mathbb{E}_{\langle\bm S_{t,i-1}\rangle}\left[\max_{v\in V}\Delta_t(v|\langle\bm S_{t,{i-1}}\rangle)\right]\nonumber\\
&=\mathbb{E}_{\langle\bm S_{t,i-1}\cup \bm S^*_{t,h-1}\rangle}\left[\max_{v\in V}\Delta_t(v|\langle\bm S_{t,{i-1}}\rangle)\right]\nonumber\\
&=\mathbb{E}_{\langle\bm S_{t,i-1}\cup \bm S^*_{t,h-1}\rangle}\left[\mathbb{E}_{\bm v^*_{t,h}}\left[\max_{v\in V}\Delta_t(v|\langle\bm S_{t,{i-1}}\rangle)\right]\right]\nonumber\\
&\geq \mathbb{E}_{\langle\bm S_{t,i-1}\cup \bm S^*_{t,h-1}\rangle}\left[\mathbb{E}_{\bm v^*_{t,h}}\left[\Delta_t(\bm v^*_{t,h}|\langle\bm S_{t,{i-1}}\rangle)\right]\right]\nonumber\\
&\geq \mathbb{E}_{\langle\bm S_{t,i-1}\cup \bm S_{t,h-1}^*\rangle}\left[\mathbb{E}_{\bm v_{t,h}^*}\left[\Delta_t(\bm v_{t,h}^*|\langle\bm S_{t,i-1}\cup \bm S_{t,h-1}^*\rangle)\right]\right]\label{lem5:ineq40}\\
&=\Delta^*_{t,h},\label{lem5:ineq4}
\end{align}
where \eqref{lem5:ineq40} holds by the adaptive submodularity (as $\langle\bm S_{t,i-1}\rangle\prec \langle\bm S_{t,i-1}\cup \bm S^*_{t,h-1}\rangle$). By combining \eqref{lem5:ineq3} and \eqref{lem5:ineq4} we obtain
\begin{equation}\label{lem5:ineq5}
\Delta^*_{t,h}\leq \mathbb{E}_{\bm S_{t,i-1}}\left[\max_{v\in V}\Delta_t(v|\bm S_{t,{i-1}})\right]\leq \Delta_{t,i}+\delta+\xi\Lambda
\end{equation}
for any $h\in [b]$. Thus, by \eqref{lem5:ineq5} we obtain
\begin{equation*}
GR^*_t(i-1,b)-GR_t(i-1)=\sum_{h\in [b]}\Delta^*_{t,h}\leq b(\Delta_{t,i}+\delta+\xi\Lambda)=b(GR_t(i)-GR_t(i-1)+\delta+\xi\Lambda),
\end{equation*}
where the first equality holds since $GR^*_t(i-1,b)-GR_t(i-1)$ is the overall expected increment achieved by adding the $b$ items selected by $\pi^*_t$ to those already selected by ${\sf \SGR}_{\delta,\xi}$. Finally, by rearranging the above inequality, we get inequality \eqref{lem5:ineq2}.

By combining \eqref{lem5:ineq1} and \eqref{lem5:ineq2} we obtain
\begin{equation}\label{lem5:ineq6}
GR_t(i)\geq \frac{GR^*_t(i-1,b)}{b}+\left(1-\frac{1}{b}\right)GR_t(i-1)-(\delta+\xi\Lambda)\geq \frac{OPT_t(b)}{b}+\left(1-\frac{1}{b}\right)GR_t(i-1)-(\delta+\xi\Lambda).
\end{equation}

By applying \eqref{lem5:ineq6} recursively from $i=1$ to $i=b$, with the base step $GR_t(0)=0$, we obtain
\begin{align*}
GR_t(b)&\geq OPT_t(b)\left(\frac{1}{b}\right)\sum_{i=0}^{b-1}\left(1-\frac{1}{b}\right)^{b}-\sum_{i=0}^{b-1}\left(1-\frac{1}{b}\right)^{b}(\delta+\xi\Lambda)\\
&\geq OPT_t(b)\left(\frac{1}{b}\right)\sum_{i=0}^{b-1}\left(1-\frac{1}{b}\right)^{b}-b(\delta+\xi\Lambda)\\
&=OPT_t(b)\left(1-\left(1-\frac{1}{b}\right)^b\right)-b(\delta+\xi\Lambda)\\
&\geq OPT_t(b)\left(1-\frac{1}{e}\right)-b(\delta+\xi\Lambda),
\end{align*}
where the last inequality follows from the known inequality $(1+x)^r< e^{rx}$, holding for any $x\geq -1$ and $r>0$, applied with $x=-1/B$ and $r=B$. By the above inequalities, the claim follows. 
\end{proof}
\begin{lemma}\label{lem6}
For any vector $\vec{b}'=(b_1',\ldots, b_T')$ of non-negative integers such that $\sum_{t\in [T]}b_t'=B$ we have $$GR(\vec{b}')-(3\delta+2\xi\Lambda)B\leq GR(\vec{b}),$$ where $\vec{b}$ is the vector returned by ${\sf BudgetGr}_{\delta,\epsilon}$. 
\end{lemma}
\begin{proof}
For any $t\in [T], i\in [n]$, let $\Delta_{t,i}$ denote the expected increment of the greedy single-round policy ${\sf \SGR}_{\delta,\xi}$ with budget $n$ at round $t$, achieved when selecting the $i$-th item. 
By exploiting the adaptive submodularity as in \eqref{lem5:ineq5}, we obtain 
\begin{equation}\label{lem6:ineq0}
\Delta_{t,j}+\delta+\xi\Lambda\geq \Delta_{t,i}
\end{equation}
for any $i,j\in [n]$ with $i\leq j$ (that is, the approximate version of inequality \eqref{lem4:ineq3} of Lemma \ref{lem4}). 

Let $\tilde{\bm \Delta}_{t,i}$ be the estimation of $\Delta_{t,i}$ considered in ${\sf BudgetGr}_{\delta,\epsilon}$ and computed by {\sf Oracle2}, so that the event $A:=$``$|\tilde{\bm \Delta}_{t,i}-{\Delta}_{t,i}|\leq \delta$ for any $t\in [T]$ and $i\in [n]$'' is true with probability at least $1-\xi$, and let $\overline{\bm \Delta}_{t,i}:=\min_{i'\in [i]}\tilde{\bm \Delta}_{t,i'}$. 
Given a realization $(\tilde{ \Delta}_{t,i})_{t,i}\sim (\tilde{\bm \Delta}_{t,i})_{t,i}$ such that event $A$ is true, $t\in [T]$, $i\in [n]$ and $i^*\in \arg\min_{i'\in [i]} \tilde{ \Delta}_{t,i'}$ (that is, $\tilde{ \Delta}_{t,i^*}=\overline{ \Delta}_{t,i}$), we have  
\begin{equation}\label{lem6:ineq11}
{\Delta}_{t,i}\leq\Delta_{t,i^*}+\delta+\xi\Lambda\leq \tilde{\Delta}_{t,i^*}+\delta+\delta+\xi\Lambda=\overline{\Delta}_{t,i}+2\delta+\xi\Lambda,
\end{equation}
where the first inequality holds by \eqref{lem6:ineq0} and the second one holds since event $A$ is true by assumption; furthermore, we have 
\begin{equation}\label{lem6:ineq1}
\overline{\Delta}_{t,i}\leq \tilde{\Delta}_{t,i}\leq\Delta_{t,i}+\delta,
\end{equation}
where the second inequality again holds since event $A$ is true by assumption.

%

Now, let $\vec{b}'=(b_1',\ldots, b_T')$ be an arbitrary vector of non-negative integers such that $\sum_{t\in [T]}b_t'$. As $\overline{\Delta}_{t,i}$ is non-decreasing in $i$, by using the same proof arguments as in Lemma \ref{lem4}, we have that the vector $\vec{b}$ returned by ${\sf BudgetGr}_{\delta,\epsilon}$ satisfies 
\begin{equation}\label{lem6:ineq20}
\sum_{t\in [T]}\sum_{i\in [b_t']}\overline{\Delta}_{t,i}\leq \sum_{t\in [T]}\sum_{i\in [b_t]}\overline{\Delta}_{t,i}.
\end{equation}
By exploiting \eqref{lem6:ineq11}, \eqref{lem6:ineq1}, \eqref{lem6:ineq20} and the definition of event $A$, we have\footnote{We observe that $(\Delta_{t,i})_{t,i}$ is a constant that does not depend on the realization of random variable $(\tilde{\bm \Delta}_{t,i})_{t,i}$.}
\begin{align}
GR(\vec{b}')&=\sum_{t\in [T]}\sum_{i\in [b_t']}\Delta_{t,i}\nonumber\\
&=\mathbb{P}[A]\cdot \mathbb{E}_{(\overline{\bm \Delta}_{t,i})_{t,i}}\left[\sum_{t\in [T]}\sum_{i\in [b_t']}\Delta_{t,i}\Bigg|A\right]+(1-\mathbb{P}[A])\cdot \mathbb{E}_{(\tilde{\bm \Delta}_{t,i})_{t,i}}\left[\sum_{t\in [T]}\sum_{i\in [b_t']}\Delta_{t,i}\Bigg|\overline{A}\right]\nonumber\\
&\leq \mathbb{P}[A]\cdot \mathbb{E}_{(\overline{\bm \Delta}_{t,i})_{t,i}}\left[\sum_{t\in [T]}\sum_{i\in [b_t']}(\overline{\bm \Delta}_{t,i}+2\delta+\xi\Lambda)\Bigg|A\right]+\xi\cdot \sum_{t\in [T]}\sum_{i\in [b_t']}\Lambda\label{lem6:ineq2}\\
&= \mathbb{P}[A]\cdot \mathbb{E}_{(\overline{\bm \Delta}_{t,i})_{t,i}}\left[\sum_{t\in [T]}\sum_{i\in [b_t']}(\overline{\bm \Delta}_{t,i})\Bigg|A\right]+B(2\delta+\xi\Lambda)+\xi B\Lambda\nonumber\\
&\leq \mathbb{P}[A]\cdot \mathbb{E}_{(\overline{\bm \Delta}_{t,i})_{t,i}}\left[\sum_{t\in [T]}\sum_{i\in [b_t]}(\overline{\bm \Delta}_{t,i})\Bigg|A\right]+B(2\delta+\xi\Lambda)+\xi B\Lambda\label{lem6:ineq3}\\
&\leq \mathbb{P}[A]\cdot \mathbb{E}_{(\overline{\bm \Delta}_{t,i})_{t,i}}\left[\sum_{t\in [T]}\sum_{i\in [b_t]}(\Delta_{t,i}+\delta)\Bigg|A\right]+B(2\delta+\xi\Lambda)+\xi B\Lambda\label{lem6:ineq4}\\
&= \mathbb{P}[A]\cdot \mathbb{E}_{(\overline{\bm \Delta}_{t,i})_{t,i}}\left[\sum_{t\in [T]}\sum_{i\in [b_t]}\Delta_{t,i}\Bigg|A\right]+B\delta+B(2\delta+\xi\Lambda)+\xi B\Lambda\nonumber\\
&\leq  \sum_{t\in [T]}\sum_{i\in [b_t]}\Delta_{t,i}+B\delta+B(2\delta+\xi\Lambda)+\xi B\Lambda\nonumber\\
&= GR(\vec{b})+B(3\delta+2\xi\Lambda)\label{lem6:ineq5}
\end{align}
where $\overline{A}$ denotes the complementary event of $A$, \eqref{lem6:ineq2} holds by \eqref{lem6:ineq11} and $1-\mathbb{P}[A]\leq \xi$, \eqref{lem6:ineq3} follows from \eqref{lem6:ineq20},  and \eqref{lem6:ineq4} holds by \eqref{lem6:ineq1}. By \eqref{lem6:ineq5}, the claim follows. 
\end{proof}
\begin{lemma}\label{lem7}
We have $\lambda\leq OPT$.
\end{lemma}
\begin{proof}
By definition of $\lambda$, there exist $t\in [T]$ and $v\in V$ such that $\lambda\leq \mathbb{E}_{\bm \eta_t\sim\mathcal{P}_t}[f_t(\{v\},\bm \eta_t)]$. As $OPT$ denotes the optimal adaptive value, we necessarily have $\mathbb{E}_{\bm \eta_t\sim\mathcal{P}_t}[f_t(\{v\},\bm \eta_t)]\leq OPT$. Thus, by combining the obtained inequalities, the claim of the lemma follows. 
\end{proof}
\end{document}